\documentclass[11pt]{article}
\usepackage[T1]{fontenc}
\usepackage[utf8]{inputenc}
\usepackage{lmodern}
\usepackage[margin=1in]{geometry}
\usepackage{amsmath,amssymb,amsthm,booktabs,needspace}
\usepackage{longtable,array,calc}
\usepackage{tikz}
\usetikzlibrary{arrows.meta,positioning,calc,fit,backgrounds}
\usepackage{microtype}
\usepackage[section]{placeins}
\usepackage[colorlinks=true,linkcolor=blue!45!black,citecolor=blue!45!black,urlcolor=blue!55!black]{hyperref}
\newtheorem{theorem}{Theorem}
\newtheorem{lemma}[theorem]{Lemma}

\theoremstyle{definition}
\newtheorem{definition}[theorem]{Definition}
\newtheorem{remark}[theorem]{Remark}
\theoremstyle{plain}
\newtheorem{proposition}[theorem]{Proposition}
\newtheorem{conjecture}[theorem]{Conjecture}
\newcommand{\ac}{\operatorname{ac}}
\newcommand{\Max}{\mathsf F}
\newcommand{\one}{\mathbf 1}
\title{Avoid-or-Compress Words for Reachable Images\\Missing Four or Five States}
\author{Qichao Wang\\[0.15em]
{\normalsize Hebei University of Technology}\\[0.45em]
Daoyu Dong\\[0.15em]
{\normalsize University of Electronic Science and Technology of China}}
\date{}
\hypersetup{pdftitle={Avoid-or-Compress Words for Reachable Images Missing Four or Five States},pdfauthor={Qichao Wang and Daoyu Dong}}

\begin{document}
\maketitle
\begin{abstract}
Given a set of states $S$ and a nonempty proper subset $A\subset S$,
an avoid-or-compress word either omits a state of $A$ from the image of
$S$ or reduces the size of that image. We determine the largest possible
length of a shortest such word when $S$ is a reachable image missing
four or five states. For strongly connected synchronizing automata,
the answers are $\min(n-|A|,11)$ and $\min(n-|A|,14)$, respectively,
and both are attained over a binary alphabet. The lower bounds follow
from uniform constructions. For the upper bounds, short-pair graphs
and the linear span of short images reduce the problem to finitely many
configurations. A computer-assisted classification eliminates these
configurations except for one complete bipartite graph, which is
excluded by a combinatorial argument.
\end{abstract}

\section{Introduction}
A synchronizing automaton admits a word that sends all its states to
one state. Bounds on the length of such a reset word often proceed by
compressing an image a little at a time. Avoid-or-compress arguments
offer a second way to make progress: instead of immediately reducing
an image, one may first move it away from a prescribed state or set.
The length needed for this step depends on the sizes of both sets.
Determining this dependence is the subject of Open Problem~4 in
\cite{szykula2026}.

Let $\mathcal A=(Q,\Sigma,\delta)$ be a complete deterministic automaton,
with $n=|Q|$. For $\varnothing\ne A\subsetneq S\subseteq Q$, put
\[
 \ac(S,A)=\min\{|w|:A\nsubseteq Sw\text{ or }|Sw|<|S|\},
 \qquad c(S)=\min\{|w|:|Sw|<|S|\}.
\]
We write $\Max(n,s,r)$ for the maximum of $\ac(S,A)$ over strongly
connected synchronizing automata, reachable images $S=Qv$ of size $s$,
and subsets $A\subset S$ of size $r$. The corresponding maximum over
binary automata is denoted by $\Max_{\rm bin}(n,s,r)$. Throughout,
$1\le r<s\le n$, and we set $m=n-s$ and $q=n-r$.

The avoid-or-compress lemma of Szyku\l a
\cite{szykula2018,szykula2026} and the Frankl--Pin bound
\cite{frankl1982,pin1983} imply
\begin{equation}\label{eq:universal}
 \Max_{\rm bin}(n,s,r)\le\Max(n,s,r)
 \le\min\{q,\tbinom{m+2}{2}\}.
\end{equation}
Related sharpness questions for avoiding a fixed state were studied
in~\cite{ferens2021}. Here we determine the extrema for all dense
images with $m=4$ or $m=5$.

\Needspace{9\baselineskip}
\begin{theorem}\label{thm:main}
For every admissible $n,r$,
\[
 \Max_{\rm bin}(n,n-4,r)=\Max(n,n-4,r)=\min(n-r,11)
\]
and
\[
 \Max_{\rm bin}(n,n-5,r)=\Max(n,n-5,r)=\min(n-r,14).
\]
\end{theorem}

Our proof has three components. First, a two-block binary construction
attains a linear bound in the number of missing states; smaller parameter
ranges are covered by related constructions. Second, we prove the
compression bounds $c(S)\le11$ for $(n-4)$-sets when $n\ge11$ and
$c(S)\le14$ for $(n-5)$-sets when $n\ge16$. These bounds hold for
arbitrary subsets, so the theorem also remains true if reachability
of $S$ is dropped from the definition of the extrema. Third, the upper
bounds combine structural lemmas with a finite computation: short-pair
graphs locate the omitted states, rank growth bounds the relevant
independent-set families, and compatible letter actions are tested
jointly. The one exceptional configuration has core $K_{4,4}$ and is
ruled out by an overlap argument.

Section~\ref{sec:prelim} records the basic tools and notation.
Section~\ref{sec:diagram} constructs the binary examples.
Sections~\ref{sec:method} and~\ref{sec:computation} give the finite
reduction and its computer-assisted verification, and
Section~\ref{sec:cases} completes the proof of Theorem~\ref{thm:main}.
The appendices contain the construction proofs, technical graph lemmas,
verification details and codimension-three compression argument.

\section{Preliminaries}\label{sec:prelim}
Words act on the right, and $|w|$ denotes length over $\Sigma$.
The alphabet is finite but otherwise unrestricted. Throughout, subsets
are arbitrary unless reachability is explicitly required. An automaton is strongly
connected if every state reaches every other state, and synchronizing
if $|Qw|=1$ for some word $w$.
For a subset $R$, let $\chi_R$ be its characteristic row vector.
The transition matrix $M_w$ is defined by $M_w(x,y)=1$ when $xw=y$
and $0$ otherwise; thus $M_{uv}=M_uM_v$. The row $\chi_RM_w$
counts trajectories with multiplicity, and has support $Rw$.
A word is \emph{safe for $(S,A)$} if it neither compresses $S$ nor
omits a state of $A$.

\begin{center}\small
\begin{tabular}{@{}ll@{}}\toprule
Notation & Meaning\\\midrule
$n,s,r$ & $|Q|,|S|,|A|$\\
$m,q$ & $n-s$, $n-r$\\
$c(S),\ac(S,A)$ & shortest compression and avoid-or-compress lengths\\
$P_h$ & graph of pairs merged within $h$ letters\\
$K_h,J_h$ & nonisolated and isolated vertices of $P_h$\\
$\alpha(G),\tau(G)$ & independence and vertex-cover numbers\\
$\mathcal D$ & a family of independent sets, represented by rows\\
$W,J$ & distinguished set and reservoir\\
$M_w,d(w)$ & transition matrix and rank defect $n-|Qw|$\\
$h,H,\ell$ & early layer, later layer and gap $\ell=H-h$\\
$\lambda$ & exit threshold $L-H$ in Appendix~\ref{app:finite-verification}\\
$C_j$ & maximum $c(R)$ over all $(n-j)$-sets in the given automaton\\
\bottomrule\end{tabular}
\end{center}
\begin{lemma}[Szyku\l a~\cite{szykula2018,szykula2026}]\label{lem:upper}
If $A\nsubseteq Sz$ for some word $z$, then $\ac(S,A)\le n-|A|$.
\end{lemma}

In a strongly connected synchronizing automaton with $n\ge2$, a reset word followed by a path to a state outside $A$ gives the required $z$. This gives the first bound in~\eqref{eq:universal}.
The second follows from the Frankl--Pin compression theorem.
\section{Binary lower bounds}\label{sec:diagram}
The construction below separates two obstructions: a pair potential
delays compression, while a coloring delays avoidance. Figure~\ref{fig:witness}
shows a representative example.

\begin{figure}[!htbp]
\centering
\begin{tikzpicture}[>=Latex,font=\small,
  unit/.style={draw,rounded corners=2pt,inner xsep=7pt,inner ysep=5pt},
  occupied/.style={unit,fill=blue!12},
  flow/.style={-Latex,semithick}]
\node[anchor=west,font=\bfseries] at (0,4.85) {(a) Letter $a$};
\node[occupied] at (3.1,4.05) {$\boxed{p_0},p_1,\ldots,p_4$ fixed};
\node[occupied] at (10.2,4.05) {$p_5\longleftrightarrow p_6$};
\node[occupied] (xa) at (2.2,2.95) {$x_1\to x_2\to x_3$};
\node[unit] (ya) at (7.3,2.95) {$y_1\to y_2\to y_3$};
\node[unit] (za) at (11.5,2.95) {$z$};
\draw[flow] (xa.east)--node[above,font=\scriptsize]{$x_3\mapsto y_1$}(ya.west);
\draw[flow] (ya.east)--node[above,font=\scriptsize]{$y_3\mapsto z$}(za.west);
\draw[flow] (za.south)--(11.5,2.20)--node[below,font=\scriptsize]{$z\mapsto x_1$}(2.2,2.20)--(xa.south);
\draw[gray!40] (0,1.75)--(13.5,1.75);
\node[anchor=west,font=\bfseries] at (0,1.30) {(b) Letter $b$};
\node[occupied,circle,double,inner sep=3pt] (merge) at (1.2,-0.15) {$p_6$};
\node[occupied] (pb) at (4.7,-0.15) {$p_5\to\cdots\to\boxed{p_0}$};
\node[unit] (yb) at (10.3,-0.15) {$y_1\to y_2\to y_3$};
\node[occupied] (xb) at (4.7,-1.45) {$x_1\to x_2\to x_3$};
\node[unit] (zb) at (10.3,-1.45) {$z$};
\draw[flow] (merge.east)--(pb.west);
\draw[flow] (pb.east)--node[above,font=\scriptsize]{$p_0\mapsto y_1$}(yb.west);
\draw[flow] (yb.north)--(10.3,0.65)--node[above,font=\scriptsize]{$y_3\mapsto p_6$}(1.2,0.65)--(merge.north);
\draw[flow] (xb.west)-|node[pos=0.3,below,font=\scriptsize]{$x_3\mapsto p_6$}(merge.south);
\draw[flow] (zb) edge[loop right,looseness=5] node[right,font=\scriptsize]{$b$}(zb);
\node[font=\footnotesize,anchor=west] at (0,-2.12)
  {Blue: initially occupied $S=P\cup X$.\quad Boxed $p_0$: protected.\quad Double circle: merger target.};
\end{tikzpicture}
\caption{The two letter actions for $t=3$, $p=7$, $r=1$.
Blue marks the initial image, boxes mark protected states, and the double
circle marks the merging target.}
\label{fig:witness}
\end{figure}

\subsection{The two-block construction}
\begin{theorem}\label{thm:two-block}
Put $m=n-s\ge3$ and $q=n-r\ge2m+2$.
There is a binary strongly connected synchronizing automaton with reachable
$S$ attaining $\ac(S,A)=\min\{q,3m-1\}$.
Consequently $\Max_{\rm bin}=\Max=q$ whenever $2m+2\le q\le3m-1$.

\end{theorem}
We construct the automaton as follows. Write $t=m-1\ge2$ and $p=n-2t-1\ge r+3$.
Take $P=\{p_0,\ldots,p_{p-1}\}$, $X=\{x_1,\ldots,x_t\}$,
$Y=\{y_1,\ldots,y_t\}$ and one further state $z$.
Letter $a$ interchanges the last two $P$-states, fixes the others, and
cycles $x_1,\ldots,x_t,y_1,\ldots,y_t,z$.
Letter $b$ sends $p_0$ to $y_1$, $p_i$ to $p_{i-1}$ for $i\ge1$,
increments indices in $X,Y$ below $t$, sends both $x_t,y_t$ to $p_{p-1}$,
and fixes $z$. Put $S=P\cup X$, $A=\{p_0,\ldots,p_{r-1}\}$ and $L=3t+2$.

\begin{proof}[Proof of Theorem~\ref{thm:two-block}]
The pair potential and color-cell arguments of
Appendix~\ref{app:binary-construction}
(Lemmas~\ref{lem:binary-potential}--\ref{lem:binary-attainment})
give compression and avoidance lower bounds $L$ and $q$. The words
$a^{2t}ba^tb$ and $a^{2t}bab^{p-r-1}$ attain these bounds, respectively.
Hence $\ac(S,A)=\min(q,L)$, where $L=3m-1$.
Lemma~\ref{lem:binary-structure} verifies the structural requirements.
When $q\le L$, equality in~\eqref{eq:universal} follows.
\end{proof}

\subsection{The remaining parameter range}
\begin{proposition}[Binary sharpness below the two-block threshold]\label{prop:smallq}
Let $m\ge4$ and $m+1\le q\le2m+1$. For every admissible $r$ there is a binary
strongly connected synchronizing automaton with a reachable $(n-m)$-image
$S$ and an $r$-set $A\subsetneq S$ such that $\ac(S,A)=q$.
\end{proposition}
\begin{proof}
Appendix~\ref{sec:appendix} supplies all three required constructions,
including their reaching and reset words.
The first cardinality face covers $q=m+1$; the second covers $q=m+2$.
For $m+3\le q\le2m+1$, take $g=q-m-1$ and $M=q-1$ in
Proposition~\ref{prop:gaps}. Then $g\ge2$, $M\ge2g$,
$n=r+M+1$, $s=r+g+1=n-m$, and $\ac(S,A)=M+1=q$.
\end{proof}
\section{Reduction to finite configurations}\label{sec:method}
Throughout this section the automaton is strongly connected and
synchronizing. The graph $P_h$ joins two distinct states when some word
of length at most $h$ merges them. Thus a set is independent in $P_h$
exactly when it survives compression through depth $h$, and
\begin{equation}\label{eq:cover}
 \exists T,\ |T|=n-m,\ c(T)>h
 \quad\Longleftrightarrow\quad \tau(P_h)\le m.
\end{equation}
Write $K_h$ for the nonisolated vertices of $P_h$ and
$J_h=Q\setminus K_h$. If $T$ is an independent $(n-m)$-set,
$Q\setminus T$ is a vertex cover; deleting its isolated vertices
still leaves a cover. Hence
\begin{equation}\label{eq:core-localization}
 |(Q\setminus T)\cap J_h|\le m-\tau(P_h).
\end{equation}
In particular, if $\tau(P_h)=m$, every such $T$ contains $J_h$.
The short-pair formulation is also used in the computational compression
framework of~\cite{kisielewicz2016}.

A lower bound on compression length now gives a family of large
independent sets whose union grows with the word length.

\begin{lemma}[Safety and union of actual short images]\label{lem:short-images}
Suppose $|S|=n-m\ge2$, $c(S)\ge L$, and $0\le j\le L-2$.
Every $Sw$ with $|w|\le j$ has size $n-m$ and is independent in $P_{L-1-j}$.
In a strongly connected automaton their union has at least
$\min(n,n-m+j)$ states.
\end{lemma}
\begin{proof}
A compression in the prefix, or a merging suffix of length at most
$L-1-j$, would compress $S$ before length $L$.
Let $U_i=\bigcup_{|w|\le i}Sw$. When $U_i\ne Q$, strong connectivity
gives a transition from $U_i$ to its complement, so
$|U_{i+1}|\ge|U_i|+1$. Since $|U_0|=n-m$, the union bound follows.
\end{proof}

The rank of a family of subsets means the rank over $\mathbb R$ of
its characteristic rows. The same images impose a rank condition.

\begin{lemma}[Rank of the image family]\label{lem:image-rank}
For a nonempty subset $S$ of a strongly connected synchronizing automaton, the span $V_i$ of the
trajectory rows $\chi_SM_w$ for $|w|\le i$ has dimension at least
$\min(i+1,n)$. Consequently, if these images are uncompressed members
of a family $\mathcal F$, then $\operatorname{rank}\mathcal F\ge\min(i+1,n)$.
\end{lemma}

To retain only the states that affect the later independent sets, we
separate their common isolated part from a bounded distinguished set.

\begin{definition}\label{def:configuration}
Fix layers $P_h\subseteq P_H$, and let $\mathcal D$ be the complete family
of maximum independent sets of $P_H$. A \emph{row} is the characteristic
vector of a member of $\mathcal D$.
Let $E$ consist of the isolated vertices of $P_h$ omitted by at least one
member of $\mathcal D$. The \emph{distinguished set} is $W=K_h\cup E$,
and the \emph{reservoir} is $J=Q\setminus W$. Thus every row contains $J$.
A \emph{configuration} consists of the early graph, the later row family,
and this partition of the states. A \emph{letter signature} records the
action on $W$, the reservoir entrance and exit counts, and the induced
row images.
\end{definition}

\begin{lemma}[Bound on the distinguished set]\label{lem:finite-named-support}
Suppose $|S|=n-m\ge2$, $c(S)\ge h+m+1$, $h\ge1$, and
$1\le\ell\le m-1$. Let $G=P_h$,
$t=\tau(G)$, $d=m-t$, and suppose $\alpha(P_{h+\ell})=n-m$.
For the complete maximum family $\mathcal D$ of $P_{h+\ell}$, put
$I_0=\bigcap\mathcal D$, $A_0=Q\setminus\bigcup\mathcal D$,
$a=|A_0|$, and $W=K_h\cup(Q\setminus I_0)$.
Then $a\le\ell$, $|W|\le2m+ad$; if $d\ge1$ and $\ell\le2d$,
\[
 |W|\le2m+\ell d-1.
\]
Each member of $\mathcal D$ omits at most $d$ isolated vertices of $G$.
In particular, bounds on $|W|$ for $(m,h,\ell,t)=(4,7,2,3)$ and
$(5,9,2,4)$ are $10$ and $11$, respectively; the former already
follows from $2m+ad\le10$.
\end{lemma}

The graph inequalities underlying this bound, and the degree, edge,
and rank restrictions on the early core, are proved in
Appendix~\ref{sec:reduction-appendix}; see in particular
Lemma~\ref{lem:early-core-filters}. The remaining case-specific
reductions are also recorded there. Together they give finitely many
configurations in each branch of the proof.

A configuration must also admit compatible letter actions. The next
lemma records the constraints inherited from an automaton; their
finite feasibility formulation appears in Section~\ref{sec:computation}.

\begin{lemma}[Constraints on letter signatures]\label{lem:finite-signatures}
Let $\mathcal A=(Q,\Sigma,\delta)$ be a strongly connected synchronizing
$n$-state automaton. Fix integers $h\ge1$, $H\ge h+1$, and $L>H$.
Suppose $|S|=n-m\ge2$, $c(S)\ge L$, and
$\alpha(P_H)=n-m$. Let $\mathcal D$ consist of all maximum
independent sets of $P_H$, and define the distinguished set and
reservoir relative to $G=P_h\subseteq P_H$ and $\mathcal D$.
The signatures induced by $\Sigma$ satisfy the following conditions:
\begin{enumerate}
\item For every $D\in\mathcal D$ and $f\in\Sigma$, $|Df|=n-m$
and $Df$ is independent in $P_h$. Thus $Df$ has size
$\alpha(P_h)-(m-\tau(P_h))$.
\item The pairs of $P_h$ have merging distances at most $h$.
The merging distances satisfy the one-step recurrence for $\Sigma$,
with distance zero on a merged pair.
\item The union of the projected transitions of all letters on the
distinguished states and the reservoir is strongly connected.
\item If $H=h+2$, then $|Df^2|=n-m$ and $Df^2$ is independent
in $P_h$ for every $D\in\mathcal D$ and $f\in\Sigma$.
The two applications use the same action on the distinguished states.
\end{enumerate}
Moreover, $S\in\mathcal D$. The configuration therefore admits a
selection of signatures with these properties and with the row and
pair distances induced by $\mathcal A$.
\end{lemma}

The proofs of the rank, support, and signature lemmas are given in
Appendix~\ref{sec:reduction-appendix}. The compression estimate with
three missing states, used to determine the possible cover numbers,
is proved in Appendix~\ref{app:codimension-three}.
\section{Computer-assisted verification}\label{sec:computation}
The finite-support reductions associate a putative counterexample
with one of finitely many configurations. For each configuration $\mathcal C$,
we test a predicate $\Pi_0(\mathcal C)$ expressing compatible row and
pair distances, safe images and their rank bounds, and connectivity of
the selected alphabet. The applicable square and auxiliary-source
conditions are included. The exact predicate and complete
configuration enumeration appear in
Appendix~\ref{app:finite-verification}.

\begin{proposition}[Finite verification]\label{prop:finite-verification}
Consider all early graphs satisfying the structural reductions and,
in each lower-cover branch, all closed later row families on the
bounded distinguished set. In an equal-cover branch use the complete
early maximum-independent-set family. Include both possible common
auxiliary-source choices where required.
The feasibility checks give the exclusions in Table~\ref{tab:coverage}.
The only remaining configuration has early core $K_{4,4}$ and the four
rows in Figure~\ref{fig:short-pair-cover}; it is excluded by
Lemma~\ref{lem:multipartite}.
\end{proposition}

An actual counterexample supplies its letter signatures and true shortest
distances, hence satisfies $\Pi_0$. The computation enumerates all such
finite configurations and signature pools. Infeasibility therefore
excludes a counterexample. Where a weaker proper-core exit test already
excludes a configuration, it also excludes $\Pi_0$.

\begin{table}[!htbp]
\centering\small
\begin{tabular}{@{}llrrp{.25\linewidth}@{}}
\toprule
$m$ & Branch & Inputs & Infeasible & Remaining step\\
\midrule
4 & $\tau(P_7)=4$ & 91 cores & --- & proper-core exit\\
4 & $\tau(P_7)=3$ & 10,142 & 10,142 & none\\
5 & $\tau(P_9)=5$ & 2,294 cores & 2,294 & none\\
5 & $\tau(P_9)=4$, $N=11$ & 3,595 & 3,595 & none\\
5 & $\tau(P_9)=4$, $5\le N\le10$ & 3,206,988 & 3,206,987 & Lemma~\ref{lem:multipartite}\\
\bottomrule
\end{tabular}
\caption{Finite verification after structural filtering; $N$ is the
early nonisolated support size. The first branch is resolved by its
proper-core exit test after joint distance enumeration.}
\label{tab:coverage}
\end{table}

\begin{remark}
The additional pair-entrance condition of
Remark~\ref{lem:pair-entrance} excludes the exceptional configuration
in a separate check.
\end{remark}

The verification compares two complete signature-generation orders,
bounded graph and row-family enumerations against literal enumeration,
and the distance encoding against an independent bounded calculation.
The verification uses the versioned implementation, solver outputs,
and bounded comparisons of independently implemented checks.

\section{The cases of four and five missing states}\label{sec:cases}
Throughout this section the automaton is strongly connected and synchronizing.
We apply the reductions to prove the upper bounds. The binary
examples in Section~\ref{sec:diagram} then give equality.

\subsection{Four missing states}\label{sec:m4}
Let $n\ge11$, $|S|=n-4$, and suppose that $c(S)\ge12$.
The Frankl--Pin bound gives $C_2\le6$, while
Theorem~\ref{lem:codimension-three} gives $C_3\le9$.
Together with~\eqref{eq:cover}, these imply
\[
 \tau(P_7)\in\{3,4\},\qquad \tau(P_9)=4.
\]
If $\tau(P_7)=4$, the images of $S$ under words of length at most
four are maximum independent sets of $P_7$, and their union is $Q$.
Lemma~\ref{lem:seven-layer-reduction} bounds the core size, degree,
edge count and row rank. This is the first branch of
Table~\ref{tab:coverage}, which the finite verification excludes.

If $\tau(P_7)=3$, use the complete maximum-independent-set family of
$P_9$. Lemmas~\ref{lem:finite-named-support} and~\ref{lem:row-capacity}
bound its support and number of rows; Lemma~\ref{lem:image-rank}
gives row rank at least three. At most two states are absent from
its union. The resulting configurations form the second branch of
Table~\ref{tab:coverage}. They are excluded by the joint distance,
connectivity and square conditions. In both branches the support is
proper because $n\ge11$.

It follows that $c(S)\le11$. When $n\le10$, the bound $q\le9$
in~\eqref{eq:universal} suffices. The matching examples are furnished
by Proposition~\ref{prop:smallq} for $5\le q\le9$ and
Theorem~\ref{thm:two-block} for $q\ge10$.

\subsection{Five missing states}\label{sec:m5}
Let $n\ge16$, $|S|=n-5$, and suppose that $c(S)\ge15$.
The codimension-three bound gives
$\tau(P_9)\in\{4,5\}$. In the cover-five case, short images of
depth at most five are maximum independent sets whose union is $Q$.
Lemma~\ref{lem:nine-layer-reduction} therefore applies. Each letter
has rank defect at most three: for defect four,
Lemma~\ref{lem:rank-descent} would give
\[
 c(S)\le 2^{5-4}(1+4+2)-1=13,
\]
and larger defects give smaller bounds. The required slow row has exit distance at least six. In particular,
$Sf^j$ is independent in $P_9$ for every letter $f$ and $0\le j\le5$.
These are the cover-five configurations excluded in
Table~\ref{tab:coverage}.

Suppose instead that $\tau(P_9)=4$. The result just proved for four
missing states implies $\alpha(P_{11})=n-5$. Let $\mathcal D$ be its
complete maximum-independent-set family. Its row rank is at least
four, at most two states are absent from its union, and its finite
support has size at most eleven. Each row is either obtained by
deleting one point from an early maximum independent set or is an
inclusion-maximal near-maximum independent set.

For an eleven-point core, Lemma~\ref{lem:auxiliary-source} supplies
two early maximum sets $R,T$ differing at one state and a single
source $R_0\in\{R,T\}$ that is independent in $P_{10}$.
Every letter maps $R_0$ into $\{R,T\}$. The corresponding
configurations are excluded by the eleven-point branch of the
finite verification. For smaller cores, Lemma~\ref{lem:entrance-bound}
and the joint tests leave only the configuration in
Figure~\ref{fig:short-pair-cover}.

The following lemma excludes this last possibility.
\begin{lemma}[Multipartite overlap exclusion]\label{lem:multipartite}
Let $h\ge1$ and $t\ge2$. Suppose a noncomplete $P_h$ in a synchronizing automaton is a complete multipartite graph on $V=V_1\sqcup\cdots\sqcup V_t$ and has isolated remaining states $J$. Assume each $V_i$ admits a family $\mathcal R_i$ of nonempty subsets whose union is $V_i$, whose intersection graph is connected, and for which every $R\cup J$ with $R\in\mathcal R_i$ is independent in $P_{h+1}$. Then this configuration is impossible.
\end{lemma}
\begin{proof}
A $P_h$ nonedge cannot be merged in one letter. Fix a letter $f$. It is therefore injective on each part and on $J$, and $Jf$ is disjoint from $Vf$. Since every $(R\cup J)f$ is $P_h$-independent, if $Jf\subseteq J$, then $Jf=J$ and $Rf$ is contained in some single part. Intersecting $R,R'$ have intersecting images, so the connected intersection graph forces all of $V_if$ into one part. Every nonedge remains a nonedge.

Otherwise $Jf$ meets some part $V_j$. Independence of $(R\cup J)f$ forces $Rf\subseteq V_j\cup J$ for every $R$, hence $Vf\subseteq V_j\cup J$ and $Jf\subseteq V_j\cup J$. Again every nonedge remains a nonedge. This holds for every letter. But a synchronizing automaton with noncomplete $P_h$ must have some letter sending a nonedge to an edge of $P_h$: follow the shortest merging word for a pair at distance greater than $h$ until the distance first drops from $h+1$ to $h$. Contradiction.
\end{proof}

In the exceptional configuration, the two intersecting triples in each
part cover that part, and their extensions by $J$ are independent in
$P_{11}$. Lemma~\ref{lem:multipartite}, applied with $h=9$, rules it out.
This proves $c(S)\le14$. For $n\le15$, the bound $q\le14$ suffices.
The examples in Proposition~\ref{prop:smallq} and
Theorem~\ref{thm:two-block} attain the upper bound for every remaining
parameter value, completing the proof of Theorem~\ref{thm:main}.

\begin{figure}[htbp]\centering
\begin{tikzpicture}[scale=.75,every node/.style={circle,draw,inner sep=2pt}]
\foreach \i in {0,1,2,3} {\node (l\i) at (0,-\i) {$\i$};}
\foreach \j in {0,1,2,3} {\pgfmathtruncatemacro{\v}{\j+4}\node (r\j) at (3,-\j) {$\v$};}
\foreach \i in {0,1,2,3} {\foreach \j in {0,1,2,3} {\draw[gray!60] (l\i)--(r\j);}}
\node[draw=none,rectangle,align=left] at (8,-1.5) {$R_1=\{0,1,2\}\cup J$\\[5pt]$R_2=\{0,1,3\}\cup J$\\[5pt]$R_3=\{4,5,6\}\cup J$\\[5pt]$R_4=\{4,5,7\}\cup J$};
\end{tikzpicture}
\caption{The exceptional configuration: core $K_{4,4}$ and its four later
rows. Every vertex of the isolated set $J$ belongs to every row.}
\label{fig:short-pair-cover}
\end{figure}
\section{Conjecture and open problems}\label{sec:conj}
The binary lower-bound constructions give
$\Max_{\rm bin}(n,n-m,r)\ge\min(q,3m-1)$ for every $m\ge4$,
and Theorem~\ref{thm:main} establishes equality for $m=4,5$.
\begin{conjecture}[Linear defect formula]\label{conj:linear}
For every $m\ge4$ and all admissible $n,r$,
\[
 \Max_{\rm bin}(n,n-m,r)=\Max(n,n-m,r)=\min(n-r,3m-1).
\]
\end{conjecture}
The remaining parameter region for exact avoid-or-compress extrema
is $m\ge6$ and $q\ge3m$.

\appendix
\section{Proofs for the binary construction}\label{app:binary-construction}
We use the automaton and notation of Section~\ref{sec:diagram}:
$t=m-1$, $p=n-2t-1$, $S=P\cup X$, $A=\{p_0,\ldots,p_{r-1}\}$,
and $L=3t+2$. The first two lemmas bound compression and avoidance
separately; the last two establish attainment and the structural properties.

\begin{lemma}[Pair potential]\label{lem:binary-potential}
In the automaton above, $c(S)\ge L$.
\end{lemma}
\begin{proof}
For $\rho_i=\min\{i,p-2\}$ define a symmetric potential $\Phi$ on unordered
pairs, zero on singleton pairs, by the following table. Indices in $X,Y$
run from $1$ to $t$, and indices in $P$ from $0$ to $p-1$.
\begin{center}\small
\begin{tabular}{ll}
pair&potential\\ \hline
$P$--$P$, $P$--$X$, $X$--$X$ (distinct)&$L$\\
$\{x_i,y_j\}$, $j=i$&$t-i+1$\\
$\{x_i,y_j\}$, $j>i$&$2t-i+1$\\
$\{x_i,y_j\}$, $j<i$&$L-j$\\
$\{x_i,z\}$&$2t-i+1$\\
$\{y_i,y_j\}$, $i<j$&$L-j$\\
$\{y_i,z\}$&$2t+1$\\
$\{p_i,y_j\}$&$\min\{L,\max\{L-j,2t+3+\rho_i\}\}$\\
$\{p_i,z\}$&$\min\{L,2t+2+\rho_i\}$
\end{tabular}
\end{center}
Substitution in the two letter actions gives
\[
 \Phi(E)\le1+\Phi(Ea),\qquad \Phi(E)\le1+\Phi(Eb).
\]
At the boundaries of the index ranges, the inequalities follow as follows. The last-two-$P$ interchange preserves $\rho$;
a downward $P$-step decreases it by at most one.
For $P$--$Y$ pairs the $a$-boundary $y_t\to z$ decreases the second term
from $2t+3+\rho$ to $2t+2+\rho$.
The $b$-boundary $p_0\to y_1$ with $j<t$ has source $L-j$
and target $L-j-1$, since $j\le t-1$; with $j=t$ its target is at least $L-1$.
For $X$--$Y$, crossing $x_t$ under $a$ gives $Y$--$Y$ with potential $L-j-1$
when $j<i$; crossing under $b$ gives $P$--$Y$ of at least that potential.
The kernel pair $\{x_t,y_t\}$ has potential one.
Crossing $y_t$ in a $Y$--$Y$ pair under $a$ gives $Y$--$z$ of
potential $2t+1=L-t-1$.
For $X$--$z$, the $a$-boundary is $\{x_t,z\}\to\{y_1,x_1\}$,
with values $t+1,t$; for $Y$--$z$, every $a$-image has potential $2t$.
All remaining cases either retain potential, decrement an index by one,
or enter a type whose potential is at least the source minus one.
Clipping at $L$ cannot increase a positive difference.

Every distinct pair in $S$ has potential $L$.
Telescoping along a merging word proves its length is at least $L$;
every compression merges an initial pair, giving the claimed lower bound.

\end{proof}
\begin{lemma}[Color cells]\label{lem:binary-colors}
Every word of length less than $q$ satisfies $A\subseteq Sw$.
\end{lemma}
\begin{proof}
Use singleton $P$-cells and $t$ color cells, giving $x_i,y_i$ color $i$
and $z$ color $1$. Let $T_a$ swap the last two $P$-cells and cycle the colors,
and let $T_b$ cycle $p_0,1,\ldots,t,p_{p-1},\ldots,p_1$.
Let $C_{x,B}=1$ if state $x$ belongs to cell $B$, and $0$ otherwise. Write $T_a,T_b$ also for the corresponding permutation matrices. Both letters satisfy
\[
 M_\ell C=CT_\ell+f_z(e_1-e_2).
\]
Here $M_\ell$ is the transition matrix, $f_z$ the standard basis column supported at $z$, and $e_i$ the color row; the displayed discrepancy is a rank-one outer product.
Since $S$ has one state in every cell, induction along each word $w$ gives
$\chi_S M_w C=\mathbf1+\sum_j\mu_j(e_1-e_2)T_{w_j}$,
where $\mu_j\ge0$ counts trajectories at $z$ immediately before the corresponding suffix $w_j$.
A trajectory from $X$ needs at least $t+1$ steps to reach $z$;
any entrance from $P$ starts at $y_1$ after at least one step and then needs
at least $t$ more. Thus the first nonzero correction has position at least $t+2$.
In the prescribed cell graph, color $2$ needs $t-1$ letters to enter $P$
at $p_{p-1}$ and at least $p-r$ more to reach $A$.
The only interchange is between adjacent last $P$-points and cannot shorten
this descending path. A negative correction therefore reaches $A$ only at
length at least $(t+2)+(t-1+p-r)=p+2t+1-r=q$.
Before then every protected coordinate has multiplicity at least one.

\end{proof}
\begin{lemma}[Attaining words]\label{lem:binary-attainment}
The words $a^{2t}bab^{p-r-1}$ and $a^{2t}ba^tb$ avoid $A$ and
compress $S$, respectively. Their lengths are $q$ and $L$.
\end{lemma}
\begin{proof}
After $a^{2t}b$, the missing $P$-point is $p_{p-1}$; after a further $a$
it is $p_{p-2}$. This is the unique $b^{p-r-1}$-preimage of $p_{r-1}$.
Hence $a^{2t}bab^{p-r-1}$ omits that protected point at length $q$.
The $X\cup Y\cup\{z\}$ portion after $a^{2t}b$ is
$\{x_2,\ldots,x_t,y_1,z\}$. After $a^t$ it contains $x_t$ and $y_t$,
so $a^{2t}ba^tb$ compresses at length $L$.
The existence of $x_t$ in the preceding set uses $t\ge2$.

\end{proof}
\begin{lemma}[Strong connectivity and synchronization]\label{lem:binary-structure}
The automaton is strongly connected and synchronizing, and $S$ is a
reachable image.
\end{lemma}
\begin{proof}
Put $B=P\cup Y$, $h=p+t$ and $E=B\cup\{z\}$.
Letter $b$ cycles $B$ in the order $p_0,y_1,\ldots,y_t,p_{p-1},\ldots,p_1$,
and $Qb^t=E$. The $a$-cycle connects all other states to $B$, proving
strong connectivity.
On $E$, $d=a^{2t}b^h$ fixes $P$, cycles $y_1\to y_t\to\cdots\to y_1$,
and sends $z$ to $y_t$. Thus $e=d^t$ fixes $B$ and sends $z$ to $y_1$.
Since $Ba^{t+1}=S$, the word $b^tea^{t+1}$ reaches $S$.
The word $a^{2t+1}$ fixes $X\cup Y\cup\{z\}$ and transposes the last two
$P$-points, adjacent on the $b$-cycle. These words generate every permutation
of $B$, with inverses realized by positive powers.
Moreover $Ba=(B\setminus\{y_1\})\cup\{z\}$,
$Bab=(B\setminus\{y_2\})\cup\{z\}$, and $Babe=B\setminus\{y_2\}$.
Permutations route any pair in a current image into the kernel of $abe$;
repeated compression synchronizes $B$. Prepending $b^te$ synchronizes $Q$.

\end{proof}

\section{Structural lemmas}\label{sec:reduction-appendix}
For a graph $G$, let $N_G(X)$ be the open neighborhood of $X$ and put $N_I(X)=N_G(X)\cap I$. The support of $G$ is
its set of nonisolated vertices. We first establish the graph-theoretic
tools used in the reduction, and then prove the statements from
Section~\ref{sec:method}.
Throughout the automaton statements below, the automaton is strongly
connected and synchronizing, and $d(u)=n-|Qu|$ denotes the rank defect
of a word $u$. The graph lemmas apply to arbitrary finite graphs.

\begin{lemma}[Support bound from near-maximum independent sets]\label{lem:near-support}
Suppose every vertex of $G$ lies in an independent set of size at least
$\alpha(G)-d$, where $d\ge0$, and let $t=\tau(G)$.
The number $N$ of nonisolated vertices satisfies $N\le t(d+2)$
(the support inequality of Gitler--Valencia~\cite{gitler2010}).
If $d>0$ and $\delta=t(d+2)-N$, the number of maximum independent
sets is at most $2^{\lfloor\delta/d\rfloor}$.
\end{lemma}
\begin{proof}
Fix a maximum independent set $I$ and put $C=V(G)\setminus I$.
For each $c\in C$, choose an independent set of size at least
$|I|-d$ containing $c$, and let $X_c$ be its intersection with $C$.
Then $|N_I(X_c)|\le|X_c|+d$. For any independent $Y\subseteq C$,
maximality of $I$ gives $|N_I(Y)|\ge|Y|$.
If $A$ is a union of previously chosen $X_c$, neighborhood submodularity
and independence of $A\cap X_c$ give
\[
 |N_I(A\cup X_c)|-|A\cup X_c|
 \le |N_I(A)|-|A|+d.
\]
At most $t$ choices cover $C$, so $|N_I(C)|\le t+td$.
Every vertex of $C$ is nonisolated, and the nonisolated vertices of $I$
are precisely $N_I(C)$, proving the first assertion.

Every maximum independent set has the unique form
$(I\setminus N_I(X))\cup X$, where $X\subseteq C$ is independent
and $|N_I(X)|=|X|$. Let $U$ be the union of these tight exchange sets.
The same submodularity argument with $d=0$ gives
$|N_I(U)|-|U|\le0$. Add at most $t-|U|$ of the sets $X_c$ to cover
$C$. Thus $N-2t\le(t-|U|)d$, or $|U|d\le\delta$.
The maximum sets are indexed by subsets of $U$.
\end{proof}

\begin{lemma}[Neighborhood bound for an independent set]\label{lem:neighborhood-ratio}
Under the hypotheses of Lemma~\ref{lem:near-support}, every independent
set $I$ containing no isolated vertex satisfies
$|I|\le(d+1)|N_G(I)|$.
\end{lemma}
\begin{proof}
Suppose $|I|>(d+1)|A|$, where $A=N_G(I)$.
If the bipartite graph between $A$ and $I$ has no matching saturating $A$,
choose $Y\subseteq A$ with $|N_I(Y)|<|Y|$. Replace $I$ by
$I\setminus N_I(Y)$ and $A$ by its actual neighborhood, a subset of
$A\setminus Y$. The strict inequality persists. Repetition terminates
with a nonempty violating pair admitting a matching saturating $A$;
a nonempty set of originally nonisolated vertices cannot have empty
actual neighborhood.

There is a maximum independent set $R$ containing this $I$.
Indeed, in any maximum set $T$, replace $T\cap(I\cup A)$ by $I$.
The matching assigns distinct points of $I\setminus T$ to $T\cap A$,
so the replacement preserves maximal cardinality and is independent.
Put $C=V(G)\setminus R$. For each $a\in A\subseteq C$, choose a
near-maximum set containing $a$ and write its intersection with $C$
as $X_a=Y_a\sqcup Z_a$, with $Y_a\subseteq A$ and $Z_a\cap A=\varnothing$.
Since $Z_a$ has no neighbors in $I$, and $|N_R(Z_a)|\ge|Z_a|$,
\[
 |N_I(Y_a)|+|Z_a|\le|N_R(X_a)|\le|X_a|+d,
 \qquad |N_I(Y_a)|\le|Y_a|+d.
\]
The $Y_a$ cover $A$. The saturating matching ensures
$|N_I(Y)|\ge|Y|$ for every $Y\subseteq A$.
Neighborhood submodularity, adding at most $|A|$ sets $Y_a$, now gives
$|I|=|N_I(A)|\le(d+1)|A|$, a contradiction.
\end{proof}

\begin{lemma}[Intersection and union of maximum independent sets]\label{lem:max-family}
Let $\mathcal D$ be all maximum independent sets of a graph $H$,
all of size $s$. Put $I_0=\bigcap\mathcal D$, $U_0=\bigcup\mathcal D$,
and $B=U_0\setminus I_0$. Then $|I_0|+|U_0|\ge2s$.
If equality holds, $H[B]$ has a perfect matching, and each member
of $\mathcal D$ selects exactly one endpoint of each matching edge,
together with all of $I_0$.
\end{lemma}
\begin{proof}
Fix $R\in\mathcal D$. Every member is
$(R\setminus N_R(X))\cup X$ for a tight independent exchange
$X\subseteq V(H)\setminus R$, with $|N_R(X)|=|X|$.
Let $A$ be the union of the tight exchanges and $D=N_R(A)$.
Submodularity as in Lemma~\ref{lem:near-support} gives $|D|\le|A|$.
Since $I_0=R\setminus D$ and $U_0=R\cup A$, the inequality follows.

In the equality case $|D|=|A|$. Process the tight exchanges successively.
Their intermediate union $A'$ has nonpositive neighborhood deficit
$|N_R(A')|-|A'|$, which cannot increase; its final value is zero.
Consequently every intermediate deficit is zero. When adding $X$,
equality forces
\[
 N_R(A')\cap N_R(X)=N_R(A'\cap X),\qquad
 |N_R(A'\cap X)|=|A'\cap X|.
\]
Every subset of the independent set $X$ satisfies Hall's inequality,
by maximality of $R$. A matching of $X$ onto $N_R(X)$ therefore exists,
and it matches $A'\cap X$ onto its equal-size neighborhood.
Its remaining edges extend any already constructed matching of
$A'$ onto $N_R(A')$. Induction gives a matching of $A$ onto $D$.
These are the two halves of $B$. Every maximum set has $|A|$ points
in $B$ and at most one on each matching edge, proving the assertion.
\end{proof}

\begin{lemma}[Saturated transversal support]\label{lem:transversal-exit}
Let $U\subsetneq Q$ be partitioned into $s\ge2$ nonempty blocks,
let $\mathcal T$ be their transversal family, and put $k=|U|-s\ge1$.
In a strongly connected synchronizing automaton, fix $S\in\mathcal T$
and suppose $Sw\in\mathcal T$ for every word $w$ with $|w|\le k$.
There is a state $z\in U$ such that every transition from $U$ to its
complement has source $z$. Every transversal avoiding $z$ maps to a
transversal under every letter. The shortest path from $S$ to $z$ has
length $k$, and $z$ lies in a nonsingleton block.

Write $a=|Q\setminus U|$, $m=k+a$, and let $J$ consist of the singleton
blocks. If $c(S)\ge h+m+1$, then
\[
 \sum_{v\notin U}|N_{P_h}(v)\cap J|\le a(a+1)/2.
\]
\end{lemma}
\begin{proof}
The reachable unions $U_i=\bigcup_{|w|\le i}Sw$ grow strictly while
proper. Since $|U_k|\le|U|=s+k$, each of the first $k$ steps adds
exactly one state, and $U_k=U$. Write $z$ for the last added state.
Every transition from $U_{k-1}$ already lands in $U$, so every exit
from $U$ has source $z$. Every short image contains every singleton
block, hence $z$ is not such a block.

Let $W$ consist of rows supported on $U$ with equal sums on all $s$
blocks. It has dimension $k+1$, and $W_z=W\cap\{x_z=0\}$ has dimension
$k$. The space $V_{k-1}$ of Lemma~\ref{lem:image-rank} lies in $W_z$
and has dimension at least $k$, hence equals $W_z$.
For any letter $f$, $V_{k-1}M_f\subseteq V_k\subseteq W$.
Thus the row of every transversal avoiding $z$ maps into $W$.
Its nonnegative integer coordinates have total $s$, so each block sum
is one and its image is again a transversal.

For $v\notin U$, let $\rho(v)$ be its distance from $z$.
A shortest such path stays outside $U$ after its first step:
returning to $U$ would require another passage through $z$ before
exiting, and that cycle could be removed. Thus $\rho(v)\le a$.
For an uncompressed $s$-set $T$, put
$\varepsilon(T)=\min_{X\in\mathcal T}|T\setminus X|$.
Changing at most one point in a transversal makes it avoid $z$.
Since function images do not increase one-sided set difference,
\[
 |Tf|=s\quad\Longrightarrow\quad
 \varepsilon(Tf)\le\varepsilon(T)+1.
\]
Reach an actual transversal containing $z$ in $k$ steps and then
follow the $\rho(v)$-step path. If compression has not occurred,
the resulting image contains $v$ and omits at most $\rho(v)$ points
of $J$. More than $\rho(v)$ neighbors of $v$ in $J$ would therefore
produce a $P_h$ edge and compression within $k+\rho(v)+h\le m+h$,
a contradiction. Hence $|N_{P_h}(v)\cap J|\le\rho(v)$.
The positive distance layers are consecutive; on $a$ vertices their
sum is at most $1+\cdots+a$, proving the stated bound.
\end{proof}

\begin{lemma}[Rank transfer]\label{lem:rank-transfer}
Let $u$ have length $a$ and rank $k\ge2$ in a synchronizing
automaton. Every $k$-set $T$ satisfies
\[
 c(T)\le a+n-k+1.
\]
\end{lemma}
\begin{proof}
Let $K_1,\ldots,K_k$ be the kernel blocks of $u$, and put
$d=n-k+1$. Assume all words of length at most $a+d$ preserve
the cardinality of $T$. For every $|v|\le d$, the $k$ trajectories
of $Tv$ visit the $k$ blocks once each; otherwise $vu$ compresses
$T$ within $a+d$ letters. Consequently the characteristic row
of $Tv$ lies in the linear subspace
\[
 H=\{x\in\mathbb R^Q:
       \sum_{y\in K_1}x_y=\cdots=\sum_{y\in K_k}x_y\},
 \qquad\dim H=n-k+1=d.
\]
Let $V_i=\operatorname{span}\{\chi_TM_v:|v|\le i\}$.
A plateau $V_i=V_{i+1}$ makes $V_i$ invariant under every
letter. It then contains the row obtained from a reset word,
which has mass $k$ concentrated in a single kernel block;
that row lies outside $H$. Hence
$\dim V_0<\dim V_1<\cdots<\dim V_d$, yielding
$\dim V_d\ge d+1$, a contradiction.
\end{proof}

\begin{lemma}[Iterated rank descent]\label{lem:rank-descent}
For $m\ge1$, $|S|=n-m\ge2$, and a word $u$
of length $a$ and defect $d=d(u)\ge1$,
\[
 c(S)\le
 \begin{cases}
 a,&d>m,\\
 2^{m-d}(a+d+2)-1,&1\le d\le m.
 \end{cases}
\]
\end{lemma}
\begin{proof}
The case $d>m$ follows from the cardinality of $Qu$.
When $d=m$, Lemma~\ref{lem:rank-transfer} applies to the
$(n-m)$-set $S$ and gives $c(S)\le a+m+1$.
For $d<m$, the proof of Lemma~\ref{lem:rank-transfer},
applied to $T=Qu$, supplies a prefix $v$ of length at most $d+1$
for which $vu$ compresses $Qu$. Thus
$Q(uvu)$ has rank at most $n-d-1$.
The word $uvu$ has length at most $2a+d+1$.
Induction on $m-d$, together with monotonicity of the
right-hand expression in the attained defect, gives
\[
 c(S)\le
 2^{m-d-1}(2a+2d+4)-1
 =2^{m-d}(a+d+2)-1.
\]

\end{proof}

\subsection*{Proof of Lemma~\ref{lem:image-rank}}
\begin{proof}
A plateau $V_i=V_{i+1}$ makes the space invariant under every letter.
A reset word followed by a path to any state $x$ supplies the row
$|S|e_x$, so the eventual span is $\mathbb R^Q$.
Starting with dimension one, the dimensions strictly increase until $n$.
\end{proof}

\subsection*{Proof of Lemma~\ref{lem:finite-named-support}}
\begin{proof}
All depth-$m$ images are independent in $G$ and cover $Q$.
Every vertex of $G$ therefore belongs to an independent set of size
$n-m=\alpha(G)-d$.
All depth-$(m-\ell)$ images belong to $\mathcal D$; their union has
at least $n-\ell$ points, so $a\le\ell$.
Put $B=(\bigcup\mathcal D)\setminus I_0$ and
$p=|K_h\cap I_0|$. Lemma~\ref{lem:max-family} gives
$|B|\le2(m-a)$. Every $G$-neighbor of $I_0$ belongs to $A_0$.
The independent-neighborhood ratio gives $p\le a(d+1)$.
Thus $|W|=a+|B|+p\le2m+ad$.

For the strict bound it remains to consider $a=\ell$ and
$|B|=2(m-a)$, since either strict inequality saves at least one.
Lemma~\ref{lem:max-family} provides a perfect matching on $B$.
Partition $U=Q\setminus A_0$ into its matching edges and the
singletons $I_0$. The family $\mathcal D$ is contained in the
transversals of this partition, whose defect is $k=m-\ell\ge1$.
All depth-$k$ actual images are transversals. Apply
Lemma~\ref{lem:transversal-exit}: $p\le a(a+1)/2$ because each
nonisolated point of $I_0$ has a neighbor in $A_0$.
Consequently
\[
 |W|\le2m+\ell(\ell-1)/2\le2m+\ell d-1
 \quad(1\le\ell\le2d).
\]
Finally the complement of each row is a cover of size $m$;
at least $t$ of its vertices are nonisolated in $G$.
\end{proof}

\begin{lemma}[Degree, size, and rank bounds for the core]\label{lem:early-core-filters}
Under $|S|=n-m\ge2$, $c(S)\ge L$, and $1\le h\le L-m-1$,
put $t=\tau(P_h)$ and $B_m=\binom{m+2}{2}$. Then
\[
 \Delta(P_h)\le m,\qquad |K_h|\le t(m-t+2),\qquad |E(P_h)|\ge h.
\]
If $0\le b=B_m+h-L\le3m-2$, then $|E(P_h)|\le b$.
When $t=m$, the maximum-independent-set row rank is at least
$\min(L-h,n)$.
\end{lemma}
\begin{proof}
The depth-$m$ images are independent in $P_h$ and cover $Q$.
An independent $(n-m)$-set containing a vertex excludes all its
neighbors, giving the degree bound. Lemma~\ref{lem:near-support}
with $d=m-t$ gives the support bound.
A short-pair layer cannot plateau before becoming complete:
its recurrence closes under one-letter preimages and direct kernels,
and a plateau would prevent any remaining pair from ever merging.
Since $S$ is independent and has at least two points, the first $h$
layers are proper and contribute at least one new edge each.

For the upper edge bound, consider any graph of maximum degree at
most $m$ and cover number at most $m$ with fewer than $3m$ edges.
It admits an $(n-m)$-set inducing exactly one edge. Indeed, take a
minimum cover $C$ of size $t'$ and its independent complement $I$.
If $c\in C$ has $d_I(c)\le m-t'+2$, retain one of its neighbors,
delete the others from $I$, and add $c$, then trim to size $n-m$.
If no such $c$ exists, $t'\le2$ contradicts the degree bound;
for $3\le t'\le m$ there are at least
$t'(m-t'+3)\ge3m$ edges, again a contradiction.
Deleting the one induced edge and repeating gives a Frankl--Pin
sequence using every edge: no later edge lies in an earlier chosen set.
Thus the maximum sequence length $p$ over any graph obeys
$p\ge\min(|E|,3m-1)$.

The Frankl--Pin refinement~\cite{kisielewicz2016} gives $c(S)\le B_m-p+h$.
If a shortest compressing word has length $L'>h$,
take its first $L'-h$ prefix images and at each prefix a pair merged
by the corresponding suffix. A later selected pair belongs to no
earlier image, since otherwise that suffix would shorten the word.
None of the images contains a pair merging within $h$ letters.
Appending a length-$p$ sequence over $P_h$ gives a Frankl--Pin
sequence of length $L'-h+p\le B_m$.
If $|E(P_h)|\ge b+1$, this implies $c(S)\le B_m-(b+1)+h=L-1$.
Finally depth-$(L-h-1)$ images are maximum independent sets when
$t=m$, and Lemma~\ref{lem:image-rank} gives their rank bound.
\end{proof}

\begin{lemma}[The cover-four branch with five missing states]\label{lem:low-cover-core-size}
Suppose $n\ge16$, $|S|=n-5$, $c(S)\ge15$, and $\tau(P_9)=4$.
Then $|K_9|\le11$. If $|K_9|=11$, the maximum-independent-set
family of $P_9$ has exactly two members differing by one state;
$P_{10}$ retains exactly one of these members, which every letter
maps into that same two-member family.
\end{lemma}
\begin{proof}
Lemma~\ref{lem:early-core-filters} first gives $|K_9|\le12$.
Suppose $|K_9|\ge11$. Since $S$ is independent in $P_{10}$,
$4\le\tau(P_{10})\le5$.
If $\tau(P_{10})=5$, apply Lemma~\ref{lem:finite-named-support}
with $m=5$, $h=9$, $\ell=1$, $d=1$ to obtain $|W|\le10$,
contradicting $K_9\subseteq W$. Thus $\tau(P_{10})=4$.
Every maximum independent set of $P_{10}$ maps under every letter,
injectively, into a maximum independent set of $P_9$.

In Lemma~\ref{lem:near-support}, the parameters $t=4$, $d=1$
give deficit $\delta=12-|K_9|$.
If $|K_9|=12$, there is a unique maximum set $R$ of $P_9$.
It is also the unique maximum set of $P_{10}$ and satisfies $Rf=R$
for every letter, contradicting strong connectivity.
For $|K_9|=11$, the proof of that lemma gives a tight-exchange union
of size at most one. Hence there are at most two maximum sets,
and if two exist they differ by one state. A unique maximum set
is excluded as above. Denote the two by $R,T$.
If both remained independent in $P_{10}$, each letter would map
both into $\{R,T\}$ and preserve their proper union of size $n-3$,
again contradicting strong connectivity. Exactly one remains;
denote it by $R_0$. Every letter maps $R_0$ into $\{R,T\}$.
\end{proof}

\begin{lemma}[Size of the later independent-set family]\label{lem:row-capacity}
Let $G\subseteq H$ be graphs with $\alpha(G)=s+1$ and
$\alpha(H)=s$. Write $\mathcal M$ for the maximum independent
sets of $G$, $\mathcal B$ for its inclusion-maximal independent
sets of size $s$, and $\mathcal D$ for the maximum family of $H$.
Then $|\mathcal D|\le|\mathcal B|+2|\mathcal M|$.
If $Z$ is the union of the $G$-isolates omitted by members of
$\mathcal D$, then $|Z|\le\max(2,|\mathcal M|)$.
\end{lemma}
\begin{proof}
A member of $\mathcal D$ is either in $\mathcal B$ or is obtained
by deleting one point from a member $R$ of $\mathcal M$.
There are at most two such deletions for a fixed $R$:
three deletions would make every pair of $R$ co-occur in an
$H$-independent set, contrary to $\alpha(H)=s$.
This proves the cardinality bound, allowing overcounting between
maximum sets.

If a member of $\mathcal D$ omits an isolated vertex $z$ of $G$, adding $z$ gives
a member of $\mathcal M$. If some $R$ supplies two such rows,
$R\setminus\{z\}$ and $R\setminus\{z'\}$, then $zz'$ is
an edge of $H$: it is the only pair of $R$ absent from both rows,
and some edge must occur in $H[R]$.
Every member of $\mathcal D$ must omit $z$ or $z'$.
Its complement, a size-$(|V|-s)$ cover of $G$, can omit at most
one $G$-isolate since $\tau(G)=|V|-s-1$.
Thus $Z\subseteq\{z,z'\}$.
Otherwise each member of $\mathcal M$ supplies at most one omitted
isolate, giving $|Z|\le|\mathcal M|$.
\end{proof}

\subsection*{Proof of Lemma~\ref{lem:finite-signatures}}
\begin{proof}
More generally, if $D$ is independent in $P_H$ and $|w|\le H-h$,
then $|Dw|=|D|$ and $Dw$ is independent in $P_h$.
A collision under $w$ would merge a pair of $D$ within $H$ letters.
A $P_h$ edge in $Dw$, preceded by $w$, would do the same.
This proves the first and fourth conditions. In the latter case the
same letter acts in both positions of its square.
The second condition is the recurrence for shortest merging words.
For the third, project paths of $\mathcal A$ onto the distinguished
states and the reservoir, omitting the reservoir vertex when
it is empty. Strong connectivity is preserved under this projection.
Finally $c(S)\ge L>H$ makes $S$ independent in $P_H$; its size
$\alpha(P_H)$ places it in $\mathcal D$.
Using the signatures of the letters and the actual shortest distances
gives the asserted selection.
\end{proof}

\begin{lemma}[Seven-layer reduction]\label{lem:seven-layer-reduction}
Under $n\ge11$, $|S|=n-4$ and $c(S)\ge12$,
the branch $\tau(P_7)=4$ is confined to nonisolated cores
with $5\le |K_7|\le8$, degree at most four and
$7\le|E(P_7)|\le10$. The maximum-independent-set family has row rank at least five.
\end{lemma}
\begin{proof}
Apply Lemma~\ref{lem:early-core-filters} with $(m,L,h)=(4,12,7)$.
It gives degree at most four, support at most eight, seven through
ten edges and row rank at least five. Cover number four forces
at least five nonisolated vertices.
\end{proof}

\begin{lemma}[Nine-layer reduction]\label{lem:nine-layer-reduction}
Under $n\ge16$, $|S|=n-5$ and $c(S)\ge15$,
\[
  \tau(P_9)\in\{4,5\}.
\]
The cover-five branch has nonisolated support between
six and ten states, at least nine early edges,
maximum degree at most five and maximum-family row
rank at least six.
\end{lemma}
\begin{proof}
The set $S$ is independent in $P_9$, giving cover number
at most five. Theorem~\ref{lem:codimension-three}
applied to all $(n-3)$-sets gives cover number at least
four. In the cover-five case, depth-five actual images
are maximum independent sets and cover all states.
Lemma~\ref{lem:early-core-filters}, with $(m,L,h)=(5,15,9)$,
gives the stated degree, support, edge and rank bounds. Complete generation of graphs of
orders six to ten with degrees one through five, followed
by the precise cover, edge and row tests, yields
the classes in Table~\ref{tab:coverage}. The corresponding joint
one-letter and distance configurations are handled in
Section~\ref{sec:computation}.
\end{proof}

\begin{lemma}[Auxiliary maximum source]\label{lem:auxiliary-source}
Suppose $n\ge16$, $|S|=n-5$, $c(S)\ge15$,
$\tau(P_9)=4$ and the nonisolated support of $P_9$
contains eleven vertices. Then $\tau(P_{10})=4$.
The maximum-independent-set family of $P_9$ consists
of two sets $R,T$ differing at one state. Exactly one
member $R_0\in\{R,T\}$ remains independent in $P_{10}$,
and every letter maps this same $R_0$
into $\{R,T\}$.
\end{lemma}
\begin{proof}
This is the eleven-vertex case of Lemma~\ref{lem:low-cover-core-size}.
$R_0$ is the unique maximum independent set of $P_{10}$.
\end{proof}

\begin{lemma}[Core entrance bound]\label{lem:entrance-bound}
Let $G=P_h$ in a synchronizing automaton, with $h\ge2$ and proper nonisolated
support $V\subsetneq Q$. If a letter $f$ sends an
isolated state into $x\in V$, then
\[
 \deg_G(x)\le d(f)+F_G(f)\le |E(G)|-h+2,
\]
where $d(f)=n-|Qf|$ and $F_G(f)$ counts distinct edges in $G$ having a
$G$-nonedge preimage under $f$.
\end{lemma}
\begin{proof}
For each neighbor of $x$ in the image $Qf$, choose a
preimage. Together with the isolated preimage of $x$,
it gives a $G$-nonedge, and its image is a $G$-edge.
At most $d(f)$ neighbors lie outside $Qf$.
The nontrivial fibers of $f$ give at least $d(f)$ distinct
kernel pairs in $P_1$. Every edge counted by $F_G(f)$ has
merging distance exactly $h$: a shorter merge, preceded by $f$,
would merge its nonedge preimage within $h$ letters.
The layers $P_i\setminus P_{i-1}$, $2\le i\le h-1$,
are nonempty. Indeed, a pair of merging distance greater than $h$
exists because $G$ is proper; the suffixes of a shortest merging
word supply every intermediate distance. These disjoint layers
contain at least $d(f)+(h-2)+F_G(f)$ edges in total.
\end{proof}

\begin{remark}[An additional entrance condition]\label{lem:pair-entrance}
Let $P_h$ be a proper short-pair layer of a synchronizing
automaton, $h\ge1$. Some letter $f$ sends a
pair outside $P_h$ to an edge of $P_h$.
\end{remark}
\begin{proof}
Choose a pair of states with minimal merging distance
$d>h$ and follow a shortest merging word.
Its successive pair distances decrease by at most
one, and the first step from distance $h+1$ reaches
distance $h$. The corresponding letter
provides the required entrance.
\end{proof}
\subsection*{Proof of Lemma~\ref{lem:upper}}
\begin{proof}
Write $r=|A|$, $s=|S|$, and let $V_i$ be the span over $\mathbb R$ of $\chi_SM_w$ for $|w|\le i$. Suppose all words of length at most $n-r$ are safe. Safety gives $|Sw|=|S|$, so $S\to Sw$ is injective: no trajectories collide, and every image state has multiplicity $1$. Their rows therefore have total $s$ and value $1$ at every coordinate of $A$. Hence $V_{n-r}$ lies in
\[
 H=\{v\in\mathbb R^Q:v_a=s^{-1}\textstyle\sum_{x\in Q}v_x\text{ for every }a\in A\}.
\]
The $r$ equations are independent since $r<s$, so $\dim H=n-r$. Now $\dim V_0=1$, and $V_i=V_{i+1}$ implies invariance under every letter and stabilization forever. Therefore some equality must occur before index $n-r$, or $\dim V_{n-r}\ge n-r+1$. Stabilization puts $\chi_SM_z$ in $H$, contradicting its zero coordinate in $A$ and total $s$.
\end{proof}
\section{Finite configurations and their verification}\label{app:finite-verification}

\subsection{The feasibility predicate}
For a configuration $\mathcal C$, let $\Pi_0(\mathcal C)$ denote the
following finite feasibility problem. Its variables select signatures
from the complete admissible pool and assign integer distances to the
retained pairs and rows. Let $\mathcal F$ be the retained equal-size
row family: the complete early maximum family in an equal-cover branch,
and the complete later maximum family otherwise. Write $H$ for its
layer index, taking $H=h$ in the former case, and put $\lambda=L-H$.
For $R\in\mathcal F$, its exit distance is
\[
 e(R)=\min\{|w|:|Rw|<|R|\ \text{or}\ Rw\notin\mathcal F\}.
\]
An exited row is represented by a terminal object of distance zero.
The slow-row condition is $\max_{R\in\mathcal F}e(R)\ge \lambda$;
the values of $\lambda$ are $5,3,6,4$ for
$(m,\tau(P_h))=(4,4),(4,3),(5,5),(5,4)$, respectively. Indeed, the initial set cannot exit before depth $L-H$.
For $\rho=\operatorname{rank}\mathcal F$ and
$\mathcal F_j=\{R\in\mathcal F:e(R)>j\}$, the safe-tail conditions are
\[
 \operatorname{rank}\mathcal F_j\ge \lambda-j\quad(0\le j<\lambda),
 \qquad
 \operatorname{rank}\mathcal F_j\le \rho-j
 \quad(\mathcal F_j\ne\varnothing, \rho<n).
\]
For the lower bound, the images of a slow row through depth $\lambda-j-1$
belong to $\mathcal F_j$, so Lemma~\ref{lem:image-rank} applies.
For the upper bound, put $V_0=\operatorname{span}\{\chi_R:R\in\mathcal F\}$
and recursively
\[
 V_{i+1}=V_0\cap\bigcap_{f\in\Sigma}\{x:xM_f\in V_i\}.
\]
Then $\chi_R\in V_j$ for $R\in\mathcal F_j$. A plateau makes the
space invariant under all letters. If it contains a vector of nonzero
coordinate sum, a reset word and paths to every state imply that it
contains every coordinate vector, contradicting $\dim V_0=\rho<n$.
If $\mathcal F_j$ is nonempty, its characteristic vector has nonzero
coordinate sum, so the first $j$ inclusions are strict. This proves
$\dim V_j\le \rho-j$.
The exit distances are positive and at most $|\mathcal F|$, or at most
$\rho$ when $\rho<n$.

Pair variables record positive merging distances at most $h$ for
edges of $P_h$, with a merged pair terminal at zero. Each level
$1,\ldots,h$ occurs. If a selected signature sends a nonedge of $P_h$
to a retained edge, that edge has distance $h$. An image pair outside
$P_h$ has distance greater than $h$, so its one-step upper inequality
places no additional restriction on a retained pair.
\begin{enumerate}
\item Selected signatures preserve the cardinality and independence
conditions on safe row images. The distance variables satisfy the
prescribed early and later pair layers, the required slow-row condition,
and the safe-tail rank bounds.
\item For every selected signature $f$ and every retained nonterminal
pair or row $x$ (writing $d_R=e(R)$ for rows), the one-step inequality
\[
 d_x\le 1+d_{xf}
\]
holds, with distance zero at a terminal object. Every nonterminal object
has a selected signature attaining equality.
\item The union of the distinguished-state and reservoir arrows of the selected
signatures is strongly connected.
\item In branches with later layer $P_{h+2}$, a letter and its square
satisfy the corresponding safe-row conditions, with consistent state
identities in both applications. In the eleven-point early-support
branch, one auxiliary source $R_0$ from
Lemma~\ref{lem:auxiliary-source} is chosen for the whole alphabet.
The two possible choices are tested separately.
\end{enumerate}
Other branch-specific power and defect bounds are those established in
the reductions. Strong connectivity is tested on the union of selected-letter arrows.
With the reservoir contracted to a vertex $*$, every distinguished vertex
both reaches $*$ and is reachable from $*$ in this union. The proper-core
exit relaxation requires each distinguished vertex to reach $*$.
Excluding a branch under this relaxation also excludes $\Pi_0$.

The lower-cover tests also retain short-image union growth. Among the
slow rows there must be an $R$ for which, at every $0\le j<\lambda$,
\[
 \left|\bigcup_{|w|\le j}Rw\right|
 \ge \min(n,n-m+j).
\]
Here words use the compatible signature pool and all these images remain
in $\mathcal F$. The reservoir is common to the rows and cancels when
the inequality is evaluated on their distinguished parts. A pool
containing the actual alphabet must satisfy this condition.

An automaton realizing $\mathcal C$ supplies a solution of
$\Pi_0(\mathcal C)$: select its signatures and use its true shortest
merging and exit distances. The shortest-distance recurrence gives
condition~(2); projection of state paths gives condition~(3).
Conditions~(1) and~(4) follow from the short-image and auxiliary-source
lemmas. Thus infeasibility excludes a counterexample. The SAT test is
one-directional: a feasible signature assignment satisfies necessary
constraints without establishing automaton realizability.

\subsection{Enumeration of configurations}
\label{sec:certificate-contract}
The support bounds make the graph generation finite. In the cover-five
branch we generate all isomorphism classes on the allowed orders and
apply the degree, edge, cover and row-rank bounds; disconnected graphs
are included. For a lower-cover core, fix a minimum cover $C$.
Its independent complement is specified by a multiset of nonempty
neighborhoods in $C$, together with an arbitrary edge set inside $C$.
Enumerating these data includes every eligible graph.

Equal-cover branches use the complete maximum-independent-set family
of the early graph. In the lower-cover branches, the candidate later
rows are the early graph's maximum
independent sets with one point deleted and its inclusion-maximal
near-maximum independent sets, including the permitted exceptional
isolates. Write these candidates as $R_1,\ldots,R_b$. For a family
$\mathcal F$, let $U(\mathcal F)$ consist of all pairs co-occurring in
its rows, including diagonal pairs, and set
\[
 \operatorname{cl}(\mathcal F)=
 \{R_i:\text{ every pair in }R_i\text{ belongs to }U(\mathcal F)\}.
\]
The complete family of later maximum independent sets is closed under
this operation. A family is rejected if its pair data also supports
an early maximum set, contradicting the prescribed later independence
number. The remaining rank, union and cardinality tests are supplied
by the finite-support reductions.

\begin{center}
\begin{minipage}{.94\linewidth}
\hrule\smallskip
\textbf{Algorithm: closed row families in a lower-cover branch}\par
Input: an ordered candidate list $R_1,\ldots,R_b$.\par
Start with $\mathrm{Visit}(\operatorname{cl}(\varnothing),1)$.\par
$\mathrm{Visit}(\mathcal F,j)$:\par
\quad If $\mathcal F$ passes all filters, emit its canonical representative.\par
\quad For $i=j,\ldots,b$ with $R_i\notin\mathcal F$:\par
\qquad Set $\mathcal F'=\operatorname{cl}(\mathcal F\cup\{R_i\})$.\par
\qquad If some $R_t\in\mathcal F'\setminus\mathcal F$ has $t<i$, skip this branch.\par
\qquad If the pair data supports an early maximum set, skip this branch.\par
\qquad Otherwise call $\mathrm{Visit}(\mathcal F',i+1)$.\par
\smallskip\hrule
\end{minipage}
\end{center}

To see completeness, fix a target closed family and at each stage add
its first missing row. Closure remains inside the target, so this path
is not rejected. Conversely, the earlier-row rule determines a unique
path to each closed family. Equivalent families are identified under
core automorphisms and a consistent permutation of exceptional isolates.

For a fixed configuration, the signature enumeration runs over the
admissible kernel partitions on distinguished states, the incoming images of
reservoir states, and injections of the kernel blocks into available
output labels. It records pair images, row images and reservoir
cardinalities. The finite signature pool covers admissible one-letter
actions for arbitrary finite alphabets.

\subsection{Exact rank computation}
The row ranks computed over $\mathbb F_{1009}$ equal rational ranks.
A zero--one minor of order at most ten has absolute determinant below
$522$, by the bordered Hadamard bound
$(k+1)^{(k+1)/2}/2^k$. At eleven distinguished vertices the only larger minor
has order eleven and all rows have weight six. Its determinant is
divisible by six, by replacing one column by the sum of all columns,
and has absolute value at most $1458<6\cdot1009$.
Thus a nonzero minor cannot vanish modulo $1009$, including in the
subfamilies used for safe-tail rank tests.

\section{Compression with three missing states}\label{app:codimension-three}
\begin{theorem}[Compression with three missing states]\label{lem:codimension-three}
In a strongly connected synchronizing automaton with $n\ge6$, every
$(n-3)$-set has a compressing word of length at most nine.
\end{theorem}

The proof proceeds by classifying the first three merging pairs. We first
establish the common reductions, then exclude the nine possible types.
Throughout this section, a string such as $012$ denotes the set
$\{0,1,2\}$, and the holes of an image are its complement in $Q$.
For a family of holes, containment of $R_i$ means containment of the
complements of the images in $R_i$. A core row lists the images of the
core states in increasing order; an asterisk denotes a free entry.
In each case, $C$ is the support of $P_3$, $H$ is its family of three-hole
covers, and $D(W)$ is the subfamily preserved, without compression, by
all tests in $W$. For a single letter we write $D_f=D(\{f\})$.
We use $[E]$ for an indicator. Products act on the right.

\subsection{Common reductions}\label{sec:c3-reductions}

Assume \(c(S)=10\), which is the only remaining case under the
Frankl--Pin bound. The refinement of \cite{kisielewicz2016} gives
\(c(S)\le10-p(P_h)+h\), where \(p(P_h)\) is the maximum length of an
\((n-3)\)-set Frankl--Pin sequence using pairs from \(P_h\). Any
\(q\le4\) distinct pairs admit a sequence of length \(q\): for its
\(i\)th pair choose three holes outside that pair, including an endpoint
outside it of each later pair, and fill the remaining holes arbitrarily.
There are at most three required endpoints and \(n-2\ge4\) available
states. The pair is in its own subset and absent from every earlier
subset. Consequently \(|P_h|\le h\) for \(h=1,2,3\). Pair layers grow
strictly until complete, so \(|P_h|=h\) at these three depths. In
particular every letter is a permutation or has defect one with the
common kernel \(01\).

Write \(R_i=\{Sw:|w|\le i\}\) and
\(V_i=\operatorname{span}\{\chi_SM_w:|w|\le i\}\). Until compression
\(R_i\subsetneq R_{i+1}\): a plateau is forward invariant and
contradicts synchronization. Also \(\dim V_i\ge\min(i+1,n)\). Indeed
a plateau in the vector spaces is invariant under all letters; a
reset word followed by a path to any state yields all rows
\(|S|e_x\). Before compression these trajectory rows equal the
characteristic rows of the actual images. Every image through depth six
is independent in \(P_3\). If \(\mathcal H\) is its three-hole cover
family and tests have maximum length \(L\), the domain preserving
\(\mathcal H\) under these tests contains \(R_{6-L}\). Such a domain
therefore has at least \(7-L\) rows. Exiting it permits compression with
at most \(L+3\) additional letters. For a permutation \(f\), the
hole action is \(H\mapsto Hf\). For a defect-one letter with missing
output \(e\), it is \(H\mapsto\{e\}\cup Hf\) when \(01\subseteq H\), and
\(H\mapsto\{e\}\cup(H\setminus\{0,1\})f\) otherwise, for the safe covers
under discussion. These formulas determine the hole families used below.

\paragraph{Rank and invariant-set obstructions.}

For the low-rank cases the following table lists the \emph{occupied core
sets}; a common true outside set \(J\) is adjoined to every row.

\begin{center}\small
\begin{tabular}{lll}
case & occupied core rows & contains\\\hline
triangle C; type 4 closed & $24,23,14,13,04,03$ & $R_5$\\
type 4 mixing & $145,135,045,035,124,123$ & $R_5$\\
type 7 \(D_{6,\mathrm{cl}}\) & $14,04,13,03,12,02$ & $R_5$\\
type 7 \(D_{6,T}\) & $135,035,125,025,134,034$ & $R_5$\\
type 7 \(D_{6,A}\) & $145,045,125,025,134,034$ & $R_5$\\
type 7 $C_A$ & $34,13,03,12,02$ & $R_4$\\
type 7 $C_V$ & $14,13,03,12,02$ & $R_4$\\
type 7 $\mathcal D_w$ & $135,035,125,025,134$ & $R_4$
\end{tabular}
\end{center}

In the first five rows, characteristic rows \(X_i\) in the displayed
order satisfy \(X_1-X_2=X_3-X_4=X_5-X_6\); these are two independent
relations. For \(C_A,C_V\), \(X_2-X_3=X_4-X_5\); for \(\mathcal D_w\),
\(X_1-X_2=X_3-X_4\). Every rank is at most four. All coefficients sum to
zero, so adjoining \(J\) preserves the relations. This contradicts
\(\dim V_4\ge5\).

The following occupied-core families each form a five-cycle on five
states: \[
\begin{array}{c|l}
\text{type 2 mixing B / closed T}&14,04,23,12,03\\
\text{type 3 mixing T}&24,04,13,03,12\\
\text{type 6 closed II}&34,14,12,03,02\\
\text{type 7 }C_T&34,14,03,12,02
\end{array}
\] Each contains \(R_4\), with the common set \(J\) adjoined. Since
every letter has defect at most one, \(\dim(V_3M_f)\ge3\). If a
point of \(J\) entered a core point \(z\), every row of \(V_3M_f\)
would be spanned by domain rows containing \(z\); there are only two,
a contradiction. Thus every letter preserves the nonempty proper
\(J\) (nonempty because \(n\ge6\)), contradicting strong connectivity.
We will refer to these two obstructions in the case analysis.

Write \(E_1=\{0,1\}\). Up to simultaneous state relabeling the ordered
three pairs, for \(n\ge 6\), have precisely these nine types:

\begin{center}\small
\begin{tabular}{clll}
\toprule
Type & $E_2$ & $E_3$ & Obstruction\\\midrule
0 & $02$ & $12$ & Four-image families; rank at most four\\
1 & $02$ & $03$ & State $0$ has no incoming path\\
2 & $02$ & $13$ & Small families; five-cycle; short exit\\
3 & $02$ & $23$ & Small families; five-cycle; short exit\\
4 & $02$ & $34$ & Six-image families of rank at most four\\
5 & $23$ & $02$ & Incompatible inverse pairs\\
6 & $23$ & $04$ & Two-step bound; five-cycle\\
7 & $23$ & $24$ & Two-step bounds; rank; five-cycle\\
8 & $23$ & $45$ & Row space of dimension at most four\\\bottomrule
\end{tabular}
\end{center}

To see that the list is complete, note that \(E_2\) shares one endpoint with \(E_1\) or is
disjoint. In the first case normalize it to \(\{0,2\}\); the ordered pair
stabilizer fixes 0,1,2, and \(E_3\) can be the remaining old edge, an
edge from each of those three vertices to one new vertex, or an edge on
two new vertices. In the disjoint case normalize \(E_2\) to \(\{2,3\}\). Its
stabilizer independently swaps the endpoints of the first two pairs.
Thus \(E_3\) joins the pairs, joins either pair to a new point, or
uses two new points. These give the five and four types above.

\subsection{Type 0: a triangle}\label{triangle-early-pair-obstruction-at-codimension-three}

Throughout \(C=\{0,1,2\}\). A hole set covers \(P_{3}\)
exactly when it is \(C\) or is \(\{i,j,x\}\), with
\(\{i,j\}\) a core pair and \(x\) outside \(C\).

\paragraph{Letter structure}\label{letter-structure}

Inverse
pairs of \(01\) lie in \(P_{2}\); inverse pairs of \(02\) lie in \(P_{3}\).
A permutation therefore preserves \(C\). A singular letter has
missing output 0 or 1. If it omits 1, inverse pairs of \(02\) force it to
map \(C\) onto \(\{0,2\}\). If it omits 0, write its kernel image
\(c=f(0)=f(1)\) and its other core image \(t=f(2)\), with \(c\ne t\)
and c,\(t\ne 0\).

Some permutation \(g\) satisfies \(g^{-1}(01)\)=\(02\), since singular
letters have no two-point preimages of \(01\). In particular \(g(02)=01\)
and \(g\) preserves \(C\) and its complement. Strong connectivity
requires a singular letter that does not preserve \(C\): otherwise
every letter preserves this proper set. Such a letter omits 0. Its
possibilities are:

\[
\begin{array}{c|cc}
 &c&t\\\hline
 A&\{1,2\}&Q\setminus C\\
 B_1&Q\setminus C&1\\
 B_2&Q\setminus C&2\\
 C&Q\setminus C&Q\setminus C
\end{array}
\]

\paragraph{Case \(A\): exclusion}\label{case-a-pure-exclusion}

For Case \(A\), the domain \(D_f\) of \(P_{3}\)-covering holes whose
\(f\)-images also cover \(P_{3}\) consists of \(C\), \(\{0,2,v\}\),
\(\{1,2,v\}\), and all \(\{0,1,x\}\), where \(v\) is the unique external
preimage of the other point in \(\{1,2\}\). Its intersection with
\(g^{-1}(D_f)\) has at most four elements. Because \(g\) sends
\(01\) to a different core pair, at most one hole \(01x\) survives its
test, namely \(x=g^{-1}(v)\); the other three holes contribute
at most three. All images reached within four letters must be in that
intersection, contradicting strict image growth.

\paragraph{Case \(B_{1}\): exclusion}\label{case-b_1-pure-exclusion}

Let \(v\) be the unique external preimage of 2. Then

\(D_f=\{C,01v\}\cup\{02x,12x:x\notin C\}.\)

Define \(D=D_f\cap g^{-1}(D_f)\cap f^{-1}(D_f)\). All images reached within four letters lie in
\(D\): appending \(f\), \(gf\) or \(ff\) must still cover \(P_{3}\) within six
letters.

The core pair \(02\) has at most one external choice in \(D\), because \(g\)
sends \(02\) to \(01\). For pair \(12\), membership in \(f^{-1}(D_f)\)
forces \(x=v\) or \(f(x)=v\), so there are at most two external choices.
The pair \(01\) has at most its single choice \(v\). The hole \(C\) belongs
to \(f^{-1}(D_f)\) only if \(c=v\). If \(c=v\), there is no
external preimage of \(v\) (it is the kernel image), so the pair \(12\) has
only its choice \(v\). Thus \(|D|\le 4\) in both cases \(c=v\) and
\(c\ne v\). Strict image growth within four letters gives a
contradiction.

\paragraph{Case \(B_{2}\): reduction to two four-image
domains}\label{case-b_2-reduction-to-two-four-image-domains}

Now \(f(2)=2\). Let \(v=f^{-1}(1)\) outside \(C\) and let
\(u=g^{-1}(v)\), also outside \(C\). Again \(D_f\) consists of
\(C\), \(\{0,1,v\}\), and all \(\{0,2,x\}\),\(\{1,2,x\}\). Put

\(D=D_f\cap g^{-1}(D_f)\cap f^{-1}(D_f)\cap(fg)^{-1}(D_f).\)

The tests correspond to \(f\),\(gf\),\(ff\),\(fgf\), all of length at most three.
Every image reached within three letters therefore lies in \(D\).

The four tests give the following restrictions:

\begin{itemize}
\item
  \(C\) is in \(D\) only if \(c=u\) (test \(fgf\));
\item
  \(\{0,1,v\}\) is in \(D\) only if \(c=v\) (test \(ff\));
\item
  core pair \(02\) has at most external choice \(u\) (test \(gf\)); it further
  requires \(u=v\) or \(f(u)=u\) (test \(fgf\));
\item
  core pair \(12\) has at most external choices \(v\) and
  \(w=f^{-1}(u)\), where \(w\) is included only if it exists
  outside \(C\) (test \(fgf\)).
\end{itemize}

If \(c=u\), \(w\) does not exist outside \(C\), since \(u\) is then the
kernel image. If \(c\ne u\), \(C\) is absent. These restrictions give
\(|D|\le 4\). Strict growth forces equality and all four domain images
to be reached within three letters. Equality has only two possibilities:

\begin{enumerate}
\item
  \(c=u=v\). Then \(D\) is \(\{012,01v,02v,12v\}\).
\item
  \(c=v\ne u\) and \(f(u)=u\). Then \(D\) is
  \(\{01v,02u,12v,12u\}\).
\end{enumerate}

Every letter must send every image in \(D\) to a \(P_{3}\)-covering
image: each was reached within three letters. In particular the letter
assumptions \(A\) and \(B_{1}\) are excluded by the arguments
above.

For the equal-parameter \(B_{2}\) domain, the occupied core rows are the
four singletons of \(C\cup\{v\}\), with a common \(J\). The four
rows are independent and contain \(R_3\). If \(J\) enters a core
point under \(f\), every row of \(V_2M_f\) is proportional to the same
occupied-core singleton row. But \(\dim(V_2M_f)\ge2\), a
contradiction. Hence \(J\) is invariant, contradicting strong
connectivity.

For the other \(B_{2}\) domain, relabel \(v=3,u=4\) and write its
occupied rows \(X_1=J\cup24\), \(X_2=J\cup13\), \(X_3=J\cup04\),
\(X_4=J\cup03\). These four rows are independent. The sole kernel
direction of a singular letter is \(k=e_0-e_1=X_4-X_2\). \(A\)
three-row subspace contains \(k\) exactly when its rows include
\(X_2,X_4\). If some letter \(f\) sends a point of \(J\) to a
core point \(z\), the output space \(V_2M_f\) is spanned by at most two
domain rows containing \(z\). Thus \(f\) is singular, \(\dim V_2=3\),
and \(R_2\) consists of \(X_2,X_4\) and one further row. At least two
possible output rows must contain \(z\). Permutations and missing-1
letters preserve \(012\), hence \(J\) can enter only \(3\) or \(4\)
under those letters. Missing-0 letters have only one available row at
\(1\) or \(2\); their rows without \(0\) are \(24,13\), disjoint, so
every entry is impossible. For a missing-1 letter entry at \(3\)
permits only \(03\). The only remaining possibility is missing-1 and
\(z=4\), with outputs \(24,04\). Its kernel image is in \(\{0,2\}\),
hence \(X_2f=X_4f\). If the third row is \(X_3\), it also contains
\(0\), so all three outputs coincide. If it is \(X_1\), then
\(X_1\cup X_4\) excludes \(1\) and \(f\) is injective there, giving
\(|X_1f\cap X_4f|=|J|\). Equal output rows instead have intersection
size \(|J|+2\), and distinct rows \(24,04\) have intersection size
\(|J|+1\). Both alternatives contradict injectivity. Thus \(J\) has no
entry into the core and is invariant, again impossible.

\paragraph{Case \(C\): a six-image
domain}\label{case-c-a-six-image-domain}

Let \(v_1\),\(v_2\) be the distinct external preimages of 1,2. The
domain \(D_f\) is exactly the six holes formed by one core pair and one
of \(v_1\),\(v_2\); \(C\) itself is not in \(D_f\). Every image reached
within five letters lies in \(D_f\), because its \(f\)-image reached
within six letters must cover \(P_{3}\). Strict growth gives all six
domain images by depth five. Every letter sends every domain image to
a \(P_{3}\)-covering image.

The six occupied rows are the triangle-\(C\) low-rank family listed
above, so the trajectory-rank contradiction excludes this case.

\subsection{Type 1: a star}\label{type-1-is-incompatible-with-strong-connectivity}

For every
letter \(f\), preimages of \(01\) are in \(P_2\), and preimages of \(02\) are in
\(P_3\).

If \(f\) is a permutation, \(f^{-1}(01)\) is \(01\) or \(02\), and
\(f^{-1}(02)\) is a different edge of the three-edge star.
Their intersection is \(\{0\}\). Since \(01\) and \(02\) also intersect at \(\{0\}\),
bijectivity gives \(f^{-1}(0)\)=0. In particular no other state
maps to 0 under \(f\).

If \(f\) is singular, let \(d\) be its unique missing output and c the
common image of 0,1. If \(d\) is not in \(01\) and c is in \(01\), the preimages
of \(01\) are the two distinct pairs \(\{0,z\}\),\(\{1,z\}\), for \(z\)
outside \(\{0,1\}\); both cannot lie in \(P_2\). If \(d\) and c are both
outside \(01\), its unique preimage pair is disjoint from \(\{0,1\}\), again
impossible in \(P_2\). Hence \(d\) is 0 or 1. If \(d=1\), examine \(02\).
When c is in \(02\), its two preimages \(\{0,z\}\),\(\{1,z\}\) cannot both
belong to \(P_3\). When c is outside \(02\), its unique preimage pair is
disjoint from \(\{0,1\}\) and cannot belong to \(P_3\). Thus \(d=0\). No state therefore maps to 0 under a singular letter.

Every letter either fixes 0 with no other incoming edge, or omits 0
entirely. State 0 is unreachable from every other state, contradicting
strong connectivity.
\subsection{Type 2: a path with opposite terminal pairs}\label{path-early-type-2-exclusion-at-codimension-three}

Set \(C=\{0,1,2,3\}\).

The three-hole covers of \(P_3\) are the four triples in \(C\), together with \(01x,03x,12x\) for \(x\notin C\). Denote this family by \(H\).

\paragraph{Compatible letters}\label{compatible-letters}

Every singular letter omits 0. Indeed its missing output must be 0 or 1,
because otherwise inverse pairs of \(01\) either are disjoint from \(01\) or
contain both 0\(z\) and 1\(z\), incompatible with \(P_{2}\). If it
omitted 1, inverse pairs of \(02\) would similarly either be disjoint from
\(01\) or contain both 0\(z\) and 1\(z\). Neither possibility belongs to
\(P_{3}\). Write such a letter \(f\) as \(f(0)=f(1)=c\), \(f(2)=t\),
\(f(3)=z\), where c,\(t\),\(z\) are distinct and nonzero.

Permutations fall into three classes, by their specified core rows:

\(A:(0,1,2,*),\qquad B:(0,2,1,*),\qquad T:(1,0,*,2).\)

The free image in B is 3 or outside \(C\); that in \(T\) is \(T_{2}=3\)
or outside \(C\). Call these closed and mixing respectively. There is
a B permutation: the new inverse pair \(02\) of \(01\) must arise from a
permutation, and class \(T\) has inverse pair \(01\) instead. There is also
a \(T\) permutation: only it has inverse pair \(13\) of \(02\). Singular
letters omit 0, so create neither of those new inverse pairs. The
classification follows directly by intersecting the inverse pairs of \(01\)
and \(02\); their common inverse endpoint is the preimage of 0.

\paragraph{Permutation mixing}\label{permutation-mixing}

The test families have the following bounds.

\begin{center}
\begin{tabular}{llcc}
\toprule
Letters & Tests & $L$ & $|D(W)|\le$\\\midrule
$B,T$ mixing & $B,T$ & 1 & 5\\
$B$ mixing, $T$ closed & $B,TB$ & 2 & 5\\
$B$ closed, $T$ mixing & $B,T,TT$ & 2 & 4\\
$B,T$ closed, $A$ mixing & $A,B,TA$ & 2 & 4\\\bottomrule
\end{tabular}
\end{center}

In the first row only core holes \(012,023\) remain; \(v=B^{-1}(3)\) fixes the outside choice for \(01\) and
\(03\), and \(T^{-1}(3)\) fixes it for \(12\). The bound five contradicts
\(|R_{5}|\ge 6\).

In the second row the candidates are \(023,123,01v,03v,12u\), where \(v=B^{-1}(3)\) and \(u=T^{-1}(v)\). Membership of \(01v\) further requires \(T(v)=v\).
Equality five therefore forces \(u=v\), and its domain is exactly

\(D=\{023,123,01v,03v,12v\},\qquad v\notin C,\quad T(v)=v.\)

If its size is smaller than five, \(|R_{4}|\ge 5\) gives a
contradiction. The equality case is treated by the five-cycle argument above.

For the third row put \(a=T_{2}\) outside \(C\) and \(v=T^{-1}(3)\). If
\(v=a\), the possible core holes are \(012\),\(013\),\(023\) and the only outside
hole is \(12v\): \(T^{-2}(3)=2\) is in \(C\). If \(v\ne a\), only
\(012\),\(013\) remain among core holes, and at most \(12v\),\(03\)\(T^{-2}(3)\)
remain outside. Both give the stated four. In the last row \(A\)
preserves \(012\) and sends 3 outside. The B test deletes all \(01x\); the
\(A\) and \(TA\) tests restrict the other two types to the listed unique
preimages: only \(012,013,03v,12T^{-1}(v)\) remain, with \(v=A^{-1}(3)\). These two rows contradict \(|R_{4}|\ge 5\).

Thus all permutations preserve \(C\).

\paragraph{Singular mixing once permutations preserve
\(C\)}\label{singular-mixing-once-permutations-preserve-c}

Choose a closed B and \(T\). Strong connectivity requires a
singular \(f\) that does not preserve \(C\). Retain \(f(0)=f(1)=c\),
\(f(2)=t\), \(f(3)=z\). Say \(t\) or \(z\) is good when it belongs to
\(\{1,3\}\). Direct substitution gives:

\begin{center}
\begin{tabular}{llll}
\toprule
Good images among \(t,z\) &  Tests  & \(L\) &  Domain bound \\
\midrule

neither & \(B,f\) & \(1\) & \(4\)\\
exactly one & \(B,f,Tf\) & \(2\) & \(4\)\\
both & \(B,fB,TfB\) & \(3\) & \(4\)\\
\bottomrule
\end{tabular}
\end{center}

For the first row B excludes \(01x\). The holes \(023\),\(123\) fail \(f\). If c
is 1 or 3, the two other core holes may remain, but only the other
member of \(\{1,3\}\) has an outside preimage; each of \(03x\) and \(12x\)
has at most that one choice. If c is neither, no core hole remains and
each outside type has at most two choices. Both counts are at most four.

In the second row \(T\) swaps \(012\) with \(013\), so the tests \(f\),\(Tf\) retain
either both of these holes or neither. The holes \(023\),\(123\) are similarly
paired. Each outside type has at most the unique outside preimage of the
other member of \(\{1,3\}\). If all four core holes remain, c must be that
other member, so this outside preimage does not exist. Thus again the
count is at most four. The first two rows contradict respectively
\(|R_{5}|\ge 6\) and \(|R_{4}|\ge 5\).

In the last row \(\{t,z\}=\{1,3\}\); mixing forces c outside \(C\).
Let \(v=f^{-1}(2)\), \(u=T^{-1}(v)\), both outside \(C\). The
domain has exactly four holes:

\[
D=\begin{cases}\{023,123,12v,03u\},&t=1,\ z=3,\\\{023,123,03v,12u\},&t=3,\ z=1.\end{cases}
\]

Both \(023\) and \(123\) have \(f\)-image
holes \(013\), which is outside \(D\), so two distinct vertices have exit
distance one. Every vertex must exit under synchronization. In a
four-vertex graph with two distinct distance-one vertices, maximum exit
distance is at most three: a shortest distance-four path would use
distinct vertices of all four levels, leaving only one vertex at level
one. The initial vertex is in \(D\). Its exit takes at most three
letters; a failed defining test takes at most three more, and a
\(P_{3}\) merge at most three more. Thus compression has length at most
nine. This excludes the last singular mixing possibility and contradicts
strong connectivity.

\subsection{Type 3: a path with adjacent initial pairs}\label{path-early-type-3-exclusion-at-codimension-three}

Let \(C=\{0,1,2,3\}\). Three-hole sets
covering \(P_{3}\) form the family \(H\) consisting of the four triples
in \(C\) and \(02x\),\(03x\),\(12x\) for \(x\) outside \(C\).

\paragraph{Compatible letters}\label{compatible-letters-1}

The inverse pairs of \(01\) belong to \(P_{2}\), and those
of \(02\) belong to \(P_{3}\). The permutation templates are:

\(A:(0,1,2,*),\qquad B:(0,2,1,*),\qquad T:(1,*,0,2).\)

\(A\),B are closed when their last core image is 3; \(T\) is closed when
\(T_{1}=3\). Otherwise that image is outside \(C\). Singular letters
either omit 0, with \(f(0)=f(1)=c\), \(f(2)=t\), \(f(3)=z\) distinct and
nonzero, or omit 1 with one of

\[
D_0:(c,c,0,2),\qquad D_2:(c,c,2,0),\qquad c\ge3.
\]

These last letters are closed when \(c=3\), otherwise mixing. To check
completeness, a singular missing output is 0 or 1 by the
inverse-pair \(01\) test. If it omits 1, its kernel image cannot belong to
\(02\): that would create both 0\(x\) and 1\(x\) as inverse pairs, which
cannot both lie in \(P_{3}\). The inverse pair of \(02\) must therefore be
the disjoint pair \(23\), giving the two displayed orientations. The
permutation templates follow by intersecting the inverse pairs of \(01\) and
\(02\). Growth from \(P_{1}\) to \(P_{2}\) requires B or \(T\); growth from
\(P_{2}\) to \(P_{3}\) requires \(T\) or \(D_{0}\) or \(D_{2}\).

\paragraph{A mixing \(T\)
permutation}\label{a-mixing-t-permutation}

Put \(a=T_{1}\) outside \(C\) and \(v=T^{-1}(3)\). Test \(T\) and \(TT\).
If \(v\ne a\), their common domain is contained in
\(\{02v,12v,03v,03T^{-1}(v)\}\), of size at most four,
contradicting \(|R_{4}|\ge 5\). If \(v=a\), the domain is exactly

\(D=\{013,123,02v,12v,03v\}.\)

The occupied core rows are the type-3 five-cycle family listed above.
The five-cycle obstruction excludes mixing \(T\).

\paragraph{A closed \(T\)
permutation}\label{a-closed-t-permutation}

Now \(T\) on \(C\) is (1,3,0,2). Its outside action is a
permutation. For an additional letter \(g\), let \(D_{2}(T,g)\) be the
domain preserving \(H\) under all words of length at most two on
\(T\),\(g\). Direct substitution gives:

\begin{center}
\begin{tabular}{lll}
\toprule
Mixing \(g\) &  Domain is contained in  & Definition of \(v,u\)\\
\midrule

\(A\) & \(123,03v,12u\) & \(v=g^{-1}(3), u=T^{-1}(v)\)\\
\(B\) & \(012,03v,12u\) & same\\
\(D_{0}\) & \(013,03u,12v\) & same\\
\(D_{2}\) & \(023,123,12u\) & same\\
\bottomrule
\end{tabular}
\end{center}

For \(A\), the tests \(g\),\(gT\),\(Tg\) remove respectively core \(013\),\(023\),\(012\);
for B they remove \(023\),\(013\),\(123\). Their outside constraints require
\(gx=3\) for \(03x\) and \(g(Tx)=3\) for \(12x\). In the \(D_{0}\) case
\(gg\) removes \(023\),\(123\) and \(g\) removes \(012\); its \(013\) survives only if
\(g(g(0))=3\), equivalently \(g(0)=v\), while the outside restrictions give
the displayed unique preimages. In \(D_{2}\), \(g\) removes \(013\), \(Tg\)
removes \(012\), and \(gg\) removes every possible \(03x\).  Each family has size at most three, contradicting
\(|R_{4}|\ge 5\). Thus every permutation and every missing-1 letter
preserves \(C\). Strong connectivity requires a mixing missing-0
letter \(f\) with core images c,\(t\),\(z\).

Call a core output good if it is 2 or 3. If at most one of
c,\(t\),\(z\) is good, \(D_{2}(T,f)\) has at most four vertices. With
none good there are no core holes, and each outside type \(03x\),\(12x\)
has at most the two outside preimages of 2,3. With only c good, there is
at most core \(012\) and one outside choice for each type. With only \(t\)
good, core holes are among \(023\),\(123\), with at most
\(03v\),\(12\)\(T^{-1}(v)\); with only \(z\) good they are among \(013\),\(123\),
with at most \(12v\),\(03\)\(T^{-1}(v)\). Here \(v\) is the outside preimage
of the other good output. These counts are at most four.

If c and one of \(t\),\(z\) are the two good outputs, the remaining core
image is outside \(C\) (otherwise \(f\) is not mixing). All outside
holes are removed by \(f\) and \(Tf\). At least one core hole is removed by
\(fT\): if \(c=2\) it is \(013\) when \(z\) is outside, or \(012\) when \(t\) is
outside; if \(c=3\) the holes \(023\),\(123\) are removed. The domain then has
size at most three.

It remains that c is outside and \(\{t,z\}=\{2,3\}\). Set
\(v=f^{-1}(1)\) and \(u=T^{-1}(v)\), both outside \(C\). Compute
\(D_{3}(T,f)\), the domain preserving \(H\) under all words of length at
most three:

\[
D_3=\begin{cases}\{013,123,03u\},&t=2,\ z=3,\\\{023,123,03v\}\cup\{12u:f(u)=v\},&t=3,\ z=2.\end{cases}
\]

The two-step families before the last tests are

\[
\begin{cases}\{013,023,123,12v\}\cup\{03x:x\notin C\},&t=2,\ z=3,\\\{012,023,123,03v\}\cup\{12x:x\notin C\},&t=3,\ z=2.\end{cases}
\]

Requiring their \(T\)- and \(f\)-successors to remain in those domains
gives the displayed \(D_{3}\) sets directly. Their sizes must be at
least four because \(|R_{3}|\ge 4\). Only \(t=3\),\(z=2\) and \(f(u)=v\)
survive. Since \(f(v)=1\), \(u\ne v\). Relabel \(u=4\),\(v=5\) and order
the remaining domain as

\(D=\{023,123,124,035\}.\)

Here the rows are \(T=(1,3,0,2,5,\ldots)\) and \(f=(c,c,3,2,5,1,\ldots)\), with
distinct unspecified outputs and \(c\notin\{0,1,2,3,5\}\). For the holes
\(023,123,124,035\), respectively, apply \(T,TT,ff,f\). Each yields
holes \(012\), outside the domain, without intermediate compression. For
example \(124f=035\) and \(123T=023\). Thus every domain row exits
within two letters. A defining test takes at most three letters, and
merging a pair in \(P_3\) takes at most three more. Thus \(c(S)\le8\),
a contradiction.

\paragraph{No \(T\)
permutation}\label{no-t-permutation}

There must now be a B permutation and a \(D_{0}\) or \(D_{2}\)
letter. A closed \(D_{0}\) is impossible even alone: its common
domain for \(D_{0}\),\(D_{0}D_{0}\) is precisely the four core holes,
fewer than \(|R_{4}|\ge 5\). A mixing \(D_{0}\) together with B also
has a two-step domain of size at most four. To see this write
\(v=D_{0}^{-1}(3)\). The \(D_{0}\) and \(D_{0}D_{0}\) tests leave at
most core \(013\), only if \(D_{0}(v)=3\) with \(v=D_{0}(0)\), and outside
holes \(12v\),\(02v\),\(03x\) with \(x=v\) or \(x=D_{0}^{-1}(v)\).
If B is closed it removes \(02v\); if core \(013\) is present,
\(D_{0}^{-1}(v)\) is not outside, giving at most three vertices. If B is
mixing, its test forces \(x=B^{-1}(3)\) for \(03x\) and for \(02x\).
There are then at most one core hole and three outside holes. Both
counts are at most four. Thus no \(D_{0}\) letter exists, and a
\(D_{2}\) letter is necessary.

If the alphabet consisted solely of \(A\),B permutations and \(D_{2}\)
letters, the family of holes \(12x\) (\(x\ne 1\),2, including \(x=0\),3)
would be invariant. \(A\),B preserve the pair \(12\). A \(D_2\) letter
sends \(120\) to \(12c\) and any other \(12x\) to \(12\)\(D_{2}(x)\); these are
again three distinct holes. This invariant family of uncompressed
\((n-3)\)-images contradicts synchronization. There is therefore a
missing-0 letter \(f\).

Choose one \(B=b\), one \(D_{2}=d\) and such an \(f\). Let \(\beta\) be 1
when \(b\) is mixing, \(\varepsilon\) be 1 when \(d\) is mixing. In their common
\(H\)-domain, \(02x\) has at most the single outside choice
\(b^{-1}(3)\), available only when \(\beta=1\); \(03x\) has at most the
outside choice \(d^{-1}(3)\), available only when \(\varepsilon=1\). Core
\(023\) is absent if \(\beta=1\), and core \(013\) is absent if \(\varepsilon=1\).
Let \(c=f(0)=f(1)\), \(t=f(2)\), \(z=f(3)\).

First suppose \(t\) is neither 2 nor 3. If neither c nor \(z\) is good,
all core holes fail and each outside type \(02x\),\(12x\) requires fx in
\(\{2,3\}\); there are at most two \(12x\) and at most one of each other
type, totaling four. If only c is good, core holes are among \(012\),\(013\),
and all outside types require the unique outside preimage \(v\) of the
other good output. There are at most three outside holes; if all three
are present, both \(b\) and \(d\) are mixing and core \(013\) is absent,
again totaling at most four. If only \(z\) is good, core holes are among
\(013\),\(023\),\(123\); the possible \(02x\),\(03x\) are offset by the respective
absence of \(023\),\(013\), leaving at most three core/outside-extra holes plus
the one possible \(12x\). If both c,\(z\) are good, \(02x\) and \(12x\)
are impossible since neither good output has an outside preimage. The
possible \(03x\) is offset by the absence of core \(013\). In each case
the one-step domain has size at most four.

Now suppose \(t=2\) or \(t=3\) and use all tests of length at most two
on \(b\),\(d\),\(f\). Define a two-element output-target set \(V\)
by

\[
V=\begin{cases}\{1,b^{-1}(3)\},&t=2,\\\{2,w\},&t=3,\end{cases}\qquad w=\begin{cases}1,&d\text{ closed},\\d^{-1}(3),&d\text{ mixing}.\end{cases}
\]

The test \(fb\) in the first case, or \(fd\) in the second, forces fx in \(V\)
for every allowed outside \(12x\). Put \(a=[c\in V]\),
\(h=[z\in V]\). The kernel output c and core output \(z\)
have no outside preimages, so at most 2-a-\(h\) outside \(12x\)
survive. The same test requires c in \(V\) for core \(012\), and \(z\) in
\(V\) for core \(023\),\(123\). Hence the core count is at most

\[
a+h(2-\beta)+(1-\varepsilon).
\]

The last term bounds core \(013\), and the \(\beta\) term removes \(023\) under
mixing B. The other two outside types contribute at most \(\beta+\varepsilon\).
The total domain size is therefore at most

\[
\begin{aligned}a+h(2-\beta)+(1-\varepsilon)+(2-a-h)+\beta+\varepsilon\\=3+h(1-\beta)+\beta\le4.\end{aligned}
\]

This also covers the common-kernel issue in \(fd\): when \(d\) is closed,
preimages of output 3 are 0,1, but \(f\) omits 0, leaving only target 1.
When \(d\) is mixing the preimage is the unique outside state
\(d^{-1}(3)\). All targets in \(V\) are distinct and neither equals
\(t\).

Thus the common two-step domain on \(b\),\(d\),\(f\) has at most four
images, contradicting \(|R_{4}|\ge 5\). This completes the case without a \(T\) permutation.

\subsection{Type 4: an adjacent pair and a disjoint edge}\label{disjoint-tail-early-type-4-exclusion-at-codimension-three}

Write
\(C=\{0,1,2,3,4\}\) and \(G=\{3,4\}\). The three-hole sets covering
\(P_{3}\) form \(H\): the seven core holes

\[
013,014,023,024,034,123,124,
\]

and \(03x\),\(04x\) for \(x\) outside \(C\).

\paragraph{All compatible letters}\label{all-compatible-letters}

Inverse pairs of \(01\) belong to \(P_{2}\) and
inverse pairs of \(02\) belong to \(P_{3}\). Permutations have core prefix
\(A\)=(0,1,2) or B=(0,2,1); they fix 0 and preserve the pair \(12\). An
inverse pair \(34\) of \(02\) cannot occur in a permutation: its inverse
pair of \(01\) is \(01\) or \(02\) and would be disjoint from \(34\), whereas the target
pairs intersect at 0.

A singular letter omits 0 or 1, by the inverse-pair \(01\) test. If it
omits 1, its kernel image cannot be 0 or 2: that would produce both
0\(x\) and 1\(x\) as inverse pairs of \(02\), incompatible with \(P_{3}\).
Its unique inverse pair of \(02\) must therefore be \(34\). It has one of the
core rows

\[
(c,c,t,0,2),\quad(c,c,t,2,0),\qquad c,t\ge3,\quad c\ne t.
\]

A singular letter omitting 0 has unrestricted distinct nonzero
images on the common kernel and the other core states. Some missing-1
letter is necessary, because only such a letter can add the new pair
\(34\) to \(P_{3}\). Call it \(d\). Swapping 3 and 4 if needed, write
\(d\)=(c,c,\(t\),0,2) on \(C\).

\paragraph{The domain of \(d\) has six
images}\label{the-domain-of-d-has-six-images}

Let \(D\) be the holes in \(H\) whose \(d\)-images also have holes in
\(H\). The hole \(034\) always has \(d\)-image holes \(012\), so is absent. Core
holes \(013\),\(014\) survive exactly when c is in \(G\); holes \(023\),\(024\),\(123\),\(124\)
survive exactly when \(t\) is in \(G\). An outside \(03x\) or \(04x\)
survives exactly when dx is in \(G\). Among the two outputs in \(G\),
every one not occupied by c or \(t\) has one outside preimage. Thus,
writing \(a=[c\in G]\), \(b=[t\in G]\),

\[
|D|=2a+4b+2(2-a-b)=4+2b.
\]

For example, \(013\) has holes \(\{1,c,0\}\), \(014\) has \(\{1,c,2\}\);
\(023\) and \(123\) have \(\{1,t,0\}\), \(024\) and \(124\) have \(\{1,t,2\}\); outside
holes have \(\{1,0,dx\}\) or \(\{1,2,dx\}\). In each case the remaining
condition is to meet \(G\).

All images reached within five letters belong to \(D\); hence
\(|D|\ge 6\). We must have \(t\) in \(G\) and \(|D|=6\), so every
\(D\)-image is reached within five letters. In particular every letter
sends every \(D\)-image to an \(H\)-image, because the resulting depth
is at most six.

There are exactly two domain shapes:

\begin{itemize}
\item
  \(c\in G\): \(c,t\) are the two members of \(G\) and
  \(D_{\mathrm{closed}}=\{013,014,023,024,123,124\}\);
\item
  \(c\notin C\): let \(v\) be the unique outside preimage under \(d\)
  of the other member of \(G\). Then
  \(D_{\mathrm{mixing}}=\{023,024,123,124,03v,04v\}\).
\end{itemize}

Both are the type-4 families in the rank table, contradicting \(\dim V_4\ge5\).

\subsection{Type 5: a crossing edge}\label{crossing-early-type-5-exclusion}

Here \(P_h\) denotes the unordered distinct-state pairs merging in at
most \(h\) letters. No strong-connectivity, synchronization,
alphabet-size, or reachable-image assumption is needed. The labels
require \(n\ge4\).

Pair \(02\) is new at
level three, so some letter \(f\) must send it to a distinct pair in
\(P_{2}\). It cannot send it to \(01\), since that would put \(02\) in
\(P_{2}\). Thus \(02\)\(f=23\) is necessary.

If \(f\) is a permutation, its inverse pair of \(01\) is either \(01\) or \(23\)
by the definition of \(P_{2}\). Its inverse pair of \(23\) would have to be
\(02\). Inverse pairs of the disjoint target pairs \(01\) and \(23\) must be
disjoint, but \(02\) intersects both possible inverse pairs of \(01\). This is
impossible.

If \(f\) is singular, its common kernel is \(01\). For \(02\)\(f=23\), \(f(0)\) is an
endpoint of \(23\) and \(f(2)\) the other endpoint. Since \(f(1)=f(0)\), also
\(12\)\(f=23\). This would put the additional pair \(12\) in \(P_{3}\),
contradicting the three specified pairs. The two cases exhaust every
letter and exclude this type.

\subsection{Type 6: an edge attached to the kernel pair}\label{kernel-tail-early-type-6-exclusion-at-codimension-three}

Let \(C=\{0,1,2,3,4\}\). The three-hole \(P_{3}\)
covers form \(H\): core holes \(012\),\(013\),\(023\),\(024\),\(034\),\(124\),\(134\) and outside
holes \(02x\),\(03x\), \(x\) outside \(C\).

\paragraph{Complete letter
classification}\label{complete-letter-classification}

The inverse pairs of \(01\) belong to \(\{01,23\}\), and inverse
pairs of \(23\) belong to \(\{01,23,04\}\). Permutations have precisely three
forms:

\begin{itemize}
\item
  \(A\): preserve each pair \(01\),\(23\);
\item
  B: exchange those two pairs;
\item
  \(T\): map \(23\) onto \(01\) and \(04\) onto \(23\), with the image of state 1
  outside \(\{0,1,2,3\}\).
\end{itemize}

The first two have core images in \(\{0,1,2,3\}\) and arbitrary outside
images. In \(T\) the four images of 0,2,3,4 are all of \(\{0,1,2,3\}\).

A singular letter cannot have its kernel image in a target pair
whose inverse pairs are required: that would give both 0\(x\) and
1\(x\), which cannot both lie in either allowed early-pair family. If
its missing output is outside \(01\), it must map \(23\) onto \(01\). Its missing
output must then be 2 or 3, since otherwise its inverse pair of \(23\) could
be none of the permitted pairs. If it omits 0 or 1, it must map \(23\) onto
\(23\). Thus all singular letters have one of these forms:

\begin{itemize}
\item
  \(U\): missing \(e\) in \(\{0,1\}\), map \(23\) onto \(23\), kernel image c in
  \(\{1-e\}\cup\{4,5,\ldots\}\);
\item
  \(V\): missing \(e\) in \(\{2,3\}\), map \(23\) onto \(01\), kernel image c in
  \(\{5-e\}\cup\{4,5,\ldots\}\).
\end{itemize}

Their images of other states are distinct and omit \(e\) and c.~Only
\(T\) can create the new pair \(04\), so some \(T\) permutation is
necessary. Swapping labels 2 and 3, the chosen \(T\) has one of the
normalized core rows

\[
\mathrm{I}:(2,a,0,1,3),\qquad\mathrm{II}:(2,a,1,0,3),\qquad a\ge4.
\]

Call it closed when \(a=4\) and mixing otherwise.

\paragraph{Orientation II}\label{orientation-ii}

If \(a\ge 5\), put \(v=T^{-1}(4)\). The common domain for \(T\) and \(TT\)
is exactly

\[
\begin{cases}\{013,034,134\},&a=v,\\\{034,134,03w\},&a\ne v,\quad w=T^{-1}(v)\notin C.\end{cases}
\]

The one-step family is \(013\),\(023\),\(034\),\(134\),\(02v\) and all
outside \(03x\). The second test removes \(023\) and \(02v\); it retains \(013\)
only if \(a=v\), and an outside \(03x\) only if \(Tx=v\). Since
\(T_{1}=a\), exactly one of these last possibilities exists. The domain
has size three, fewer than \(|R_{4}|\ge 5\), so mixing orientation II is
impossible.

If \(a=4\), the \(T\),\(TT\) domain is precisely the five core holes

\(D=\{012,023,034,124,134\}.\)

Its occupied rows form the type-6 five-cycle listed above. The
common-entry argument excludes closed orientation II.

\paragraph{Orientation \(I\) requires a
breaker}\label{orientation-i-requires-a-breaker}

Let \(T=(2,a,0,1,3)\) on \(C\), put \(v=T^{-1}(4)\), and
\(\sigma=[a=v]\), and put \(X=\{02x:x\notin C\}\). The domain for \(T\),\(TT\) is

\[
\begin{cases}\{012,013,023,024,134\}\cup X,&T\text{ closed},\\\{012,023,024,03v\}\cup X\cup\{124:\sigma=1\},&T\text{ mixing}.\end{cases}
\]

Here \(v=1\) when closed and \(v\notin C\) when mixing; the one-step
mixing family is
\[\{012,023,024,124,03v\}\cup X.\]

For a permutation \(g\), call it a keeper when
\(\{g(0),g(2)\}=\{0,2\}\). For a singular \(g\) with missing output \(e\),
call it a keeper when \(\{e,g(2)\}=\{0,2\}\). If every letter were
a keeper, the family of holes \(02x\) for all \(x\ne 0\),2
(including \(x=1\),3,4) would be invariant. For singular letters, the
hole \(012\) maps to \(02c\), and every other \(02x\) maps to \(02g(x)\); each is
again three distinct holes. Thus these \((n-3)\)-images could never
compress, contradicting synchronization. There is a breaker \(g\).

All orientation-II letters have already been excluded, so the complete
breaker list is as follows. Below \(h\) denotes \(g(4)\), except in the \(T\)
row where \(h=g(1)\); each free output obeys the rank and missing-output
rules.

\begin{center}
\begin{tabular}{lp{0.69\linewidth}}
\toprule
Name  &  Core row or conditions \\
\midrule

\(A_{03}\) & \((0,1,3,2,h)\)\\
\(B_{03}\) & \((3,2,0,1,h)\)\\
\(A_{\mathrm{bad}}\) & \((1,0,k,l,h), \{k,l\}=\{2,3\}\)\\
\(B_{\mathrm{bad}}\) & \((k,l,1,0,h), \{k,l\}=\{2,3\}\)\\
\(T_{\mathrm{other}}\) & \((3,h,1,0,2), h\ge 4\)\\
\(U_{03}\) & \((c,c,3,2,h)\), missing 0\\
\(V_{03}\) & \((c,c,0,1,h)\), missing 3\\
\(U_{\mathrm{bad}}\) & \(U\) missing 1\\
\(V_{\mathrm{bad}}\) & \(V\) with \(g(2)=1,g(3)=0\)\\
\bottomrule
\end{tabular}
\end{center}

The list follows from the \(A\),B,\(T\),\(U\),\(V\) classification by
checking whether the carrier of the holes \(02x\) is \(02\),\(03\),\(12\) or \(13\).
Carrier \(02\) is precisely the keeper case. A carrier \(12\)
or \(13\) covers \(P_{3}\) when its third point is 0, as well as when
it is 4; both possibilities are included below.

\paragraph{Two-step bounds}\label{uniform-two-step-counting-table}

Let \(D(T,g)\) be the \(H\)-domain of all words of length at most two on
\(T\),\(g\). Its core holes are among the displayed \(T\),\(TT\) core holes.
Its outside holes are only \(02x\) and, in the mixing case, at most the
single \(03v\). Write \(\gamma_i=[c=i]\), \(\eta_i=[h=i]\). Each
triple in the tables bounds

\((N_c,N_{02},N_{03})\), where the components count core holes, outside holes of type \(02x\), and outside holes of type \(03x\), respectively.

For closed orientation \(I\), the bounds are as follows. In the \(U_{\mathrm{bad}}\) row, write \(\gamma=\gamma_0+\gamma_4\) and \(\eta=\eta_0+\eta_4\).

\begin{center}
\begin{tabular}{lp{0.69\linewidth}}
\toprule
Breaker
 & Component bounds
 \\
\midrule

\(A_{03}\) & \((4,0,0)\)\\
\(B_{03}\) & \((3,0,0)\)\\
\(A_{\mathrm{bad}}, B_{\mathrm{bad}}\) & \((3+\eta_{4},1-\eta_{4},0)\)\\
\(T_{\mathrm{other}}\) & \((2+\eta_{4},1-\eta_{4},0)\)\\
\(U_{03}\) & \((3+\gamma_{1}+\eta_{1},1-\gamma_{1}-\eta_{1},0)\)\\
\(V_{03}\) & \((1+\gamma_{2}+\eta_{2}+\eta_{4},1-\gamma_{2}-\eta_{2},0)\)\\
\(U_{\mathrm{bad}}\) & \((2\gamma+2\eta,2-\gamma-\eta,0)\)\\
\(V_{\mathrm{bad}}\) & \((3+\gamma_{4}+\eta_{4},1-\gamma_{4}-\eta_{4},0)\)\\
\bottomrule
\end{tabular}
\end{center}

Every sum is at most four. For \(U_{\mathrm{bad}}\), c,\(h\) are distinct so
\(\gamma\),\(\eta\) are each indicators, giving sum 2+\(\gamma+\eta\le 4\).

For mixing orientation \(I\), use \(\gamma=\gamma_2+\gamma_v\), \(\eta=\eta_2+\eta_v\) in the \(V_{03}\) row, and \(\gamma=\gamma_0+\gamma_4\), \(\eta=\eta_0+\eta_4\) in the \(U_{\mathrm{bad}}\) row. The resulting bounds are:

\begin{center}
\begin{tabular}{lp{0.69\linewidth}}
\toprule
Breaker
 & Component bounds
 \\
\midrule

\(A_{03}\) & \((1+\eta_v+\sigma \eta_{4},1-\eta_v,\eta_{4})\)\\
\(B_{03}\) & \((1+\eta_v+\sigma,1-\eta_v,0)\)\\
\(A_{\mathrm{bad}}\) & \((1+\eta_{4}+\sigma(1-\eta_{4}),1-\eta_{4},1-\eta_{4})\)\\
\(B_{\mathrm{bad}}\) & \((1+\eta_{4}(1+\sigma),1-\eta_{4},\eta_{4})\)\\
\(T_{\mathrm{other}}\) & \((\sigma \eta_{4},1-\eta_{4},\eta_{4})\)\\
\(U_{03}\) & \((1+\gamma_v+\eta_v(1+\sigma),1-\gamma_v-\eta_v,1)\)\\
\(V_{03}\) & \((\gamma+\eta(1+\sigma),2-\gamma-\eta,0)\)\\
\(U_{\mathrm{bad}}\) & \((\gamma+\eta(1+\sigma),2-\gamma-\eta,1)\)\\
\(V_{\mathrm{bad}}\) & \((1+\gamma_{4}+\eta_{4}(1+\sigma),1-\gamma_{4}-\eta_{4},\eta_{4})\)\\
\bottomrule
\end{tabular}
\end{center}

Here \(v\ge 5\), so the indicator sums in the \(V_{03}\) row are again
indicators. Every row sums to at most four (some have the sharper bound three).

To verify the table algebraically, apply the hole formula in Section~\ref{sec:c3-reductions}
to the core holes and the two outside types. The restrictions on outside
\(02x\) reduce to unique preimages of these target sets:

\begin{itemize}
\item
  permutation bad rows: \(\{4\}\);
\item
  \(U_{03}\): \(\{v\}\), with \(v=1\) in the closed case;
\item
  \(V_{03}\): \(\{2,v\}\), with target 1 already having core preimage 3
  when closed;
\item
  \(U_{\mathrm{bad}}\): \(\{0,4\}\);
\item
  \(V_{\mathrm{bad}}\): \(\{4\}\), since output 0 already has core preimage 3.
\end{itemize}

Every target occupied by kernel image c or core image \(h\) has no
outside preimage. This gives all the negative terms in the second
column. For permutation \(A_{03}\)/\(B_{03}\) the targets of \(gT\) on
\(02x\) are \(\{2,v\}\), but output 2 already has a core preimage,
leaving only \(v\); when closed, output \(v=1\) also has a core
preimage, leaving none.

Core counts are obtained by substituting the five/four candidate core
holes of \(T\),\(TT\) into \(g\),\(gT\),\(Tg\),\(gg\). For illustration, \(U_{03}\) has
three possible unconditional closed core holes \(013\),\(023\),\(134\) and the two
extra holes \(012\),\(024\) only if \(c=1\),\(h=1\). In the mixing case it has
unconditional \(023\), \(012\) only if \(c=v\), and \(024\) (also \(124\) if
\(\sigma=1\)) only if \(h=v\).

For \(V_{03}\) the closed hole \(013\) cannot survive: its first image would
require \(c=4\) and be \(134\), but its second image then requires \(h=4\),
contradicting distinct outputs. The remaining core candidates are
\(023\),\(012\) if \(c=2\), \(024\) if \(h=2\) and \(134\) if \(h=4\). In the mixing
case \(023\) fails \(gT\), and \(03v\) also fails \(gT\): its possible first images
\(013\) or \(134\) both fail \(T\). The displayed \(\gamma\),\(\eta\) restrictions on
\(012\),\(024\),\(124\) then give its table row. The bad \(U\) core counts simply
test whether c,\(h\) belong to \(\{0,4\}\); the bad \(V\) rows have three
unconditional possible closed core holes, with \(012\)/\(024\) requiring
\(c=4\),\(h=4\). In the mixing bad \(V\) case only \(023\) is unconditional,
\(012\) requires \(c=4\), and \(024\), the optional \(124\) and \(03v\) require
\(h=4\). The permutation rows follow from their explicit core templates
in exactly the same substitution.

It follows that \(|D(T,g)|\le4\). But every image
reached within four letters belongs to it, whereas \(|R_{4}|\ge 5\).
This excludes orientation \(I\).

\subsection{Type 7: an edge attached to the second pair}\label{second-pair-tail-early-type-7-exclusion-at-codimension-three}

Put \(C=\{0,1,2,3,4\}\). The three-hole covers \(H\) of \(P_{3}\) are the
seven core holes \(012\), \(023\), \(024\), \(123\), \(124\), \(034\), \(134\) and outside holes
\(02x\),\(12x\), \(x\) outside \(C\).

\paragraph{Complete letter
classification}\label{complete-letter-classification-1}

Inverse pairs of \(01\) must be \(01\) or \(23\); inverse pairs of
\(23\) must be \(01\),\(23\), or \(24\). The compatible permutation core rows are
precisely

\(A:(k_0,k_1,l_0,l_1,h),\quad B:(l_0,l_1,k_0,k_1,h),\quad T:(k_0,k_1,l_0,h,l_1).\)

where (\(k_0\),\(k_1\)) orders 0,1, (\(l_0\),\(l_1\)) orders 2,3, and \(h\ge 4\). Their
outside outputs are distinct and avoid the four specified outputs.

For a singular letter, write \(e\) for its omitted output and
\(c=f(0)=f(1)\). If \(e\) is outside \(01\), the inverse pair of \(01\) must be \(23\).
Its kernel image cannot lie in \(01\), since then both 0\(z\) and 1\(z\)
would be inverse pairs. If \(e\) were also outside \(23\), its inverse pair
of \(23\) would be disjoint from \(23\), hence could only be \(01\). That would put
c in \(23\) and give two forbidden inverse pairs there. Therefore \(e\) is
in \(23\). If \(e\) is in \(01\), the inverse pair of \(23\) is \(23\) or \(24\), with c
outside \(23\). This gives precisely

\[
\begin{array}{ccl}U_A&(c,c,l_0,l_1,h)&e\in\{0,1\},\\U_T&(c,c,l_0,h,l_1)&e\in\{0,1\},\\V&(c,c,k_0,k_1,h)&e\in\{2,3\}.\end{array}
\]

For \(U\), c is 1-\(e\) or at least 4; for \(V\), c is 5-\(e\) or at
least 4. The other images are distinct and avoid \(e\),c and the
specified pair. These conditions also suffice for compatibility.

\paragraph{The seven forced domains}\label{the-seven-forced-domains}

The corresponding hole families are listed below. An outside point
\(v\) or \(w\) is relabeled 5.

\begin{center}
\begin{tabular}{lll}
\toprule
name  &  hole domain  &  contains \\
\midrule

\(D_{6,\mathrm{cl}}\) & \(023,123,024,124,034,134\) & \(R_{5}\)\\
\(D_{6,T}\) & \(024,124,034,134,02v,12v\) & \(R_{5}\)\\
\(D_{6,A}\) & \(023,123,034,134,02v,12v\) & \(R_{5}\)\\
\(C_A\) & \(012,024,124,034,134\) & \(R_{4}\)\\
\(C_T\) & \(012,023,124,034,134\) & \(R_{4}\)\\
\(C_V\) & \(023,024,124,034,134\) & \(R_{4}\)\\
\(\mathcal D_w\) & \(024,124,034,134,02w\) & \(R_{4}\)\\
\bottomrule
\end{tabular}
\end{center}

Taking complements in the distinguished set gives the occupied families
above. All except \(C_T\) are excluded by rank at most four; \(C_T\) is
excluded by the five-cycle common-entry argument. The following
substitutions derive the domains exhaustively.

\paragraph{Eliminate all letters taking 2 to
3}\label{eliminate-all-letters-taking-2-to-3}

For a permutation \(f\) with \(f(2)=3\), its free core output is \(f(4)\) in
form \(A\), or \(f(3)\) in form \(T\). If that output is 4, the one-step
\(H\)-preserving domain is \(D_{6,\mathrm{cl}}\). Otherwise \(v=f^{-1}(4)\)
is outside \(C\), and the domain is \(D_{6,A}\) or \(D_{6,T}\)
respectively.

For a singular \(U\) row with \(f(2)=3\) the domain size is
\(6-[c=4]\). To see the count, let \(z\) be its free core output.
Hole \(012\) survives exactly when \(c=4\); the two holes \(023\),\(123\) (form
\(A\)) or \(024\),\(124\) (form \(T\)) survive exactly when \(z=4\). The four
other core holes survive. An outside preimage of 4 supplies two outside
holes unless 4 is already c or \(z\). Since c and \(z\) are distinct,
the count is

\[
[c=4]+2[z=4]+4+2(1-[c=4]-[z=4])=6-[c=4].
\]

If \(c=4\), size five contradicts \(|R_{5}|\ge 6\). Otherwise it is
precisely \(D_{6,\mathrm{cl}}\) when \(z=4\), or \(D_{6,A}\)/\(D_{6,T}\)
with \(v=f^{-1}(4)\) outside \(C\). The six-member low-rank families
therefore exclude every letter with \(f(2)=3\).

\paragraph{Eliminate \(V\) rows omitting
3}\label{eliminate-v-rows-omitting-3}

For \(V\) missing 3 put \(A=[c\in \{2,4\}]\), \(B=[h\in \{2,4\}]\). Its one-step domain has size 4-\(A\)+2B: hole \(012\)
contributes \(A\), holes \(023\),\(123\) fail, the other four core holes
contribute 4B, and the outside holes contribute 2(2-\(A\)-B). If \(B=0\)
it has at most four members. If \(B=1\) and \(A=1\) it has five. The
only remaining case has six members, \(c\ge 5\) and \(h\) in \(\{2,4\}\),
and is \(D_{6,T}\) with \(v\) the outside preimage of the other member
of \(\{2,4\}\). It too is excluded.

\paragraph{Eliminate B four-cycles}\label{eliminate-b-four-cycles}

After the preceding cuts a B permutation has \(B^{2}(2)=2\) or 3. In
the latter case a simultaneous exchange of labels 0,1 normalizes it
to

\(b=(2,3,1,0,a),\qquad a\ge4.\)

For \(a=4\) the \(b\),\(bb\) domain is \(C_V\), and its low-rank family
excludes this case. For \(a\ge 5\) put \(v=b^{-1}(4)\). The domain is

\[
\begin{cases}\{023,034,024\},&a=v,\\\{023,034,02w\},&a\ne v,\quad w=b^{-1}(v)\notin C.\end{cases}
\]

It has only three members, contradicting \(|R_{4}|\ge 5\).

\paragraph{\(A\) B two-cycle and a breaker are
necessary}\label{a-b-two-cycle-and-a-breaker-are-necessary}

If no B permutation remains, every letter preserves the family \(I\) of
three-hole sets containing 2 and at least one of 0,1. The remaining
\(A\),\(T\) permutations fix 2 and preserve \(01\); \(U\) maps the two
carriers to \(\{e,2\}\); \(V\) omitting 2 maps them to \(\{2,f(2)\}\). The
both-kernel-endpoints case also preserves \(I\). Images with these holes
never compress, contrary to synchronization. Choose a B two-cycle
and normalize it by exchanging 0,1 to

\(b=(2,3,0,1,a),\qquad a\ge4.\)

Call a permutation a keeper if \(\{g(0),g(2)\}=\{0,2\}\), and a
singular row a keeper if \(\{e,g(2)\}=\{0,2\}\). Every keeper
preserves the family \(F\) of all holes \(02x\), \(x\ne 0\),2. Thus
synchronization requires a breaker. After the preceding cuts the
exhaustive breaker forms are

\[
\begin{array}{ccl}P_A&(1,0,2,3,h)&\\P_T&(1,0,2,h,3)&\\P_B&(3,2,1,0,h)&\\U_A&(c,c,2,3,h)&e=1,\\U_T&(c,c,2,h,3)&e=1,\\V&(c,c,1,0,h)&e=2.\end{array}
\]

Here \(h\ge 4\) for permutations; c is 0 or \(\ge4\) for \(U\) and
3 or \(\ge4\) for \(V\). In singular rows \(h\) avoids the missing
output, c, and the specified pair.

\paragraph{Mixing B two-cycles}\label{mixing-b-two-cycles}

If \(a\ge 5\) put \(v=b^{-1}(4)\). Directly testing \(b\) and \(bb\) gives
the domain

\[
F\cup\{12v\}\cup\{034:a=v\}.
\]

In particular it has at most two members outside \(F\). For a
breaker \(g\), write \(N_c\) for the count of core holes in \(D(b,g)\) intersect
\(F\), and \(N_o\) for the count of its outside holes. The following upper
bounds follow by mapping the carriers of \(02x\) under \(g\) and then
under \(b\). For \(U\) set \(\gamma=[c\in \{0,v\}]\), \(\eta=[h\in \{0,v\}]\); for \(V\) set \(\gamma=[c=v]\),
\(\eta=[h=v]\); for \(P_A\),\(P_T\) set \(\eta=[h=v]\).

\begin{center}
\begin{tabular}{lll}
\toprule
breaker  &  \(N_c\)  &  \(N_o\) \\
\midrule

\(P_A,P_T\) & \(1+\eta\) & \(1-\eta\)\\
\(P_B\) & \(0\) & \(1-[h=4]\)\\
\(U_A,U_T\) & \(\gamma+\eta\) & \(2-\gamma-\eta\)\\
\(V\) & \(1+\gamma+\eta\) & \(1-\gamma-\eta\)\\
\bottomrule
\end{tabular}
\end{center}

For example \(U\) converts the carriers to \(12\); after \(b\) these are \(03\).
The third output must be 0 or \(v\) so that its \(b\)-image is 2 or 4.
The outputs already occupied by c or \(h\) remove precisely \(\gamma+\eta\)
available outside preimages. The analogous \(V\) carrier is \(12\) with
output 0 already on input 3, leaving only the target \(v\). The bounds
sum to at most two; adding the at most two non-\(F\) members gives
\(|D(b,g)|\le4\), contradicting \(|R_{4}|\ge 5\).
Hence \(a=4\).

\paragraph{Closed B two-cycles}\label{closed-b-two-cycles}

For \(a=4\) the \(b\),\(bb\) domain is \(F\cup\{124,034,134\}\). Count
core and outside \(02x\) members of \(D(b,g)\) as \(N_c,N_o\). For \(U\) put
\(\gamma=[c\in \{0,4\}]\), \(\eta=[h\in \{0,4\}]\),
\(\eta_{0}=[h=0]\), \(\eta_{4}=[h=4]\); for \(V\) put
\(\gamma_{4}=[c=4]\), \(\eta_{4}=[h=4]\).

\begin{center}
\begin{tabular}{lll}
\toprule
breaker  &  \(N_c\)  &  \(N_o\) \\
\midrule

\(P_A,P_T\) & \(1+4\eta_{4}\) & \(1-\eta_{4}\)\\
\(P_B\) & \(1+3\eta_{4}\) & \(1-\eta_{4}\)\\
\(U_A\) & \(\gamma+2\eta_{0}+4\eta_{4}\) & \(2-\gamma-\eta\)\\
\(U_T\) & \(\gamma+\eta_{0}+3\eta_{4}\) & \(2-\gamma-\eta\)\\
\(V\) & \(3+\gamma_{4}+2\eta_{4}\) & \(1-\gamma_{4}-\eta_{4}\)\\
\bottomrule
\end{tabular}
\end{center}

These bounds are obtained by testing the six core holes individually and
counting outside preimages of the same carrier targets. Their sums are
at most five; \(P_B\) and \(U_T\) have sums at most four. For \(U\) the
cases \(h=0\) and \(h=4\) must be distinguished: \(h=0\) does not
preserve \(034\),\(134\).

Since \(|R_{4}|\ge 5\), equality is necessary. Testing its core holes
gives:

\begin{itemize}
\item
  \(P_A\) has \(h=4\) and domain \(C_A\);
\item
  \(P_T\) has \(h=4\) and domain \(C_T\);
\item
  \(U_A\) has \(h=4\). If \(c=0\) its domain is \(C_A\); otherwise
  \(c\ge 5\) and its domain is \(\mathcal D_w\), where \(w=g^{-1}(0)\) is
  outside \(C\);
\item
  \(V\) has \(h=4\) and domain \(C_V\).
\end{itemize}

All these five-member domains are excluded by the rank and common-entry
arguments above. This excludes the remaining letters.

\subsection{Type 8: three disjoint pairs}

Here every safe image selects one vertex from each of \(01,23,45\) and
contains every state outside those six. The three block-coordinate sums
are equal, and each outside coordinate equals their common sum. The row
space has dimension at most four and contains no reset row \(|S|e_x\)
(also when \(n=6\) and the outside set is empty). The reset-row version
of the plateau argument thus gives \(\dim V_4\ge5\), whereas \(R_4\)
lies in this space. This is a contradiction. The nine cases exhaust
the possible early pair layers, proving
Theorem~\ref{lem:codimension-three}.

\section{Constructions for the smaller values of \texorpdfstring{$q$}{q}}\label{sec:appendix}
\begin{theorem}\label{thm:sharp-faces}
The following bounds hold, and every lower bound has a binary witness.
\begin{enumerate}
\item $\Max(n,r+1,r)=n-r$ for $n\ge2$ and $1\le r\le n-1$.
\item $\Max(n,r+2,r)=n-r$ for $n\ge4$ and $1\le r\le n-3$.
\item For $r\ge1$, $s-r\ge3$ and $n\ge2s-r-1$, one has
\[
 \Max(n,s,r)=n-r.
\]
\end{enumerate}
For $1\le r<n$, one has $\Max(n,n,r)=1$.
\end{theorem}
Thus the first two cardinality faces are determined, including their full-image boundary. For every fixed $(s,r)$ with $s-r\ge3$, the bound is sharp at every $n\ge2s-r-1$, in particular at $n=r+k(s-r-1)+2$ for $k\ge2$.

The upper bound follows from Lemma~\ref{lem:upper}; when $S=Q$,
a nonpermutation letter compresses immediately.

We will also use the following elementary fact. If $t$ cycles $x_0,x_1,\ldots,x_{s-1}$ and $e$ sends $x_0$ to $x_1$ while fixing the other points, then
\begin{equation}\label{eq:unit}
 e(t^{s-1}e)^{s-2}
\end{equation}
maps this set to $x_1$. Indeed, after $e$ the image is $\{x_1,\ldots,x_{s-1}\}$; each $t^{s-1}e$ deletes the last point of the remaining interval. This also covers $s=2$.

\Needspace{8\baselineskip}
\subsection{The first two cardinality faces}
\begin{proposition}\label{prop:one}
Part~(1) of Theorem~\ref{thm:sharp-faces} holds.
\end{proposition}
\begin{proof}
Put $L=n-r\ge1$ and $Q=\{0,\ldots,n-1\}$. Define
\[
 xa=\begin{cases}x+1&x<L,\\x&x\ge L,\end{cases}
 \qquad S=\{0,L,L+1,\ldots,n-1\},\qquad A=S\setminus\{0\}.
\]
Let $b$ cycle $S$ in the displayed order and send every point outside $S$ to $L$. Then $Qb=S$. The $a$-path from $0$ to $L$ and the $b$-cycle prove strong connectivity.

For $0\le j<L$, set $S_j=A\cup\{j\}$. One has $S_jb=S_0$, and $S_ja=S_{j+1}$ for $j<L-1$; $S_{L-1}a=A$. Induction shows that every word of length less than $L$ has image $S_j$ with $j<L$, so is safe. The word $a^L$ compresses $S$, proving $\ac(S,A)=L$.

On $S$, $c=a^L$ is the unit $0\mapsto L$. Formula~\eqref{eq:unit} with $t=b$ shows that $c(b^rc)^{r-1}$ resets $S$. It resets $Q$ as well, since $Qc=A=Sc$. This proves synchronization, also when $L=1$ or $r=1$.
\end{proof}

\begin{proposition}\label{prop:two}
Part~(2) of Theorem~\ref{thm:sharp-faces} holds.
\end{proposition}
\begin{proof}
Put $M=n-r\ge3$. Take disjoint sets $A=\{p_0,\ldots,p_{r-1}\}$ and $U=\{u_1,\ldots,u_M\}$, and let $S=A\cup\{u_1,u_2\}$. Letter $a$ fixes $A$, while
\[
 p_ib=p_{i+1}\quad(0\le i<r-1),\qquad p_{r-1}b=u_1.
\]
For odd $M$, let $a$ cycle $u_1,\ldots,u_M$ and set
\[
 u_jb=\begin{cases}p_0&j\text{ odd},\\u_2&j\text{ even}.\end{cases}
\]
For even $M$ (necessarily $M\ge4$), set
\[
 u_1a=u_2,\quad u_ja=u_{j+1}\ (2\le j<M),\quad u_Ma=u_2,
 \qquad
 u_jb=\begin{cases}u_2&j=1\text{ or }j\ge3\text{ odd},\\p_0&j\text{ even}.\end{cases}
\]
In both cases $Qb=S$. From $p_0$ the $b$-path reaches $u_1$, whose $a$-orbit covers $U$; some state of $U$ returns to $p_0$ under $b$. Every state therefore reaches and is reached from $p_0$.

The two moving points of $Sa^j$ remain distinct for all $j\ge0$. For $0\le j\le M-2$, their $b$-images are $p_0$ and $u_2$, so $Sa^jb=S$. At $j=M-1$, the pair is $\{u_M,u_1\}$ for odd $M$ and $\{u_M,u_2\}$ for even $M$; both points map to $p_0$. Thus $h=a^{M-1}b$ compresses $S$. Every word shorter than $M$ splits into $a$-runs followed by $b$, each run having length at most $M-2$, and a trailing $a$-run. All these transitions are safe. Hence $\ac(S,A)=M$.

Set $s=r+2$, and $t=ab$ for odd $M$, $t=b$ for even $M$. On $S$, $t$ cycles $p_0,\ldots,p_{r-1},u_1,u_2$. If $e$ is the unit $u_1\mapsto u_2$ on $S$, then $h=et$ there. Consequently
\[
 h(t^{s-2}h)^{s-2}=e(t^{s-1}e)^{s-2}t
\]
on $S$, and this is a reset by~\eqref{eq:unit}. Prepending $b$ gives a reset of $Q$.
\end{proof}

\subsection{Every larger fixed gap}
\begin{proposition}\label{prop:gaps}
For $r\ge1$, $g\ge2$ and $M\ge2g$, there is a binary strongly connected synchronizing automaton with
\[
 n=r+M+1,\quad |S|=s=r+g+1,\quad |A|=r,
 \qquad \ac(S,A)=M+1.
\]
The subset $S$ is an actual reachable image.
\end{proposition}
\begin{proof}
Take disjoint sets $P=\{p_0,\ldots,p_r\}$ and $U=\{u_1,\ldots,u_M\}$.
Set $A=P\setminus\{p_r\}$ and $S=P\cup\{u_1,\ldots,u_g\}$.
First suppose $M\ge2g+1$ and put $h=1+((M-1)\bmod g)$.
Then $1\le h\le g$ and $M\ge2g+h$.
Let $a$ fix $P$, send $u_j$ to $u_{j+1}$ for $j<M$, and send $u_M$ to $u_{h+1}$.
Define $p_0b=u_h$, $p_ib=p_{i-1}$ for $1\le i\le r$, and
\[
 u_jb=\begin{cases}
 p_r&j\equiv h-1\pmod g,\\
 u_1&j=g,\ h\ge2,\\
 u_{h+g}&j=M,\\
 u_{j+1}&\text{otherwise}.
 \end{cases}
\]
The cases are disjoint. On $S$, $b$ cycles
$p_0,u_h,\ldots,u_g,u_1,\ldots,u_{h-1},p_r,\ldots,p_1$.
The segment $u_1,\ldots,u_{h-1}$ is empty for $h=1$.
Ordinary outside states follow increasing indices to the next index
congruent to $h-1$, then enter $p_r$.
The exceptional path of length $g+1$ is
\[
u_M\longrightarrow u_{h+g}\longrightarrow\cdots
\longrightarrow u_{h+2g-1}\longrightarrow p_r.
\]
Its indices exist because $M\ge2g+h$.
Thus $Qb^{g+1}=S$. The $a$-path from $u_1$ covers $U$, and the $b$-cycle covers $P$,
proving strong connectivity.

\emph{Safety.}
Use singleton $P$-cells and color cells $1,\ldots,g$, assigning
$u_j$ color $1+((j-h-1)\bmod g)$.
Let $C$ be the $n\times s$ state-to-cell incidence matrix.
Let $T_a$ fix $P$ and cycle the colors, and let $T_b$ cycle
$p_0,g,1,\ldots,g-1,p_r,\ldots,p_1$.
For a standard state column $f_q$ and cell row $e_c$, one has
\begin{equation}\label{eq:error}
 M_aC=CT_a,\qquad M_bC=CT_b+f_{u_M}(e_g-e_1).
\end{equation}
The last term is an outer product of an $n\times1$ column and a $1\times s$ row.
These cells are bookkeeping devices, not a quotient automaton.
For $w=\ell_1\cdots\ell_t$, iteration gives
\begin{equation}\label{eq:expand}
 \chi_SM_wC=\one_s+
 \sum_{\substack{1\le i\le t\\\ell_i=b}}
 \mu_i(e_g-e_1)T_{\ell_{i+1}\cdots\ell_t},
 \qquad \mu_i=(\chi_SM_{\ell_1\cdots\ell_{i-1}})_{u_M}\ge0.
\end{equation}
This counts trajectories even after collisions, with products in word order.
A trajectory from $S$ needs at least $M-g$ letters to reach $u_M$:
initial indices are at most $g$, every transition within $U$ increases an
index by at most one, and entrances from $P$ have index $h\le g$.
Hence a nonzero correction has $i\ge M-g+1$.
In the cell graph, color $g$ needs at least $g-1$ letters to reach color $g-1$,
and color $1$ needs at least $g$ letters to reach $A$.
The latter route traverses $1,\ldots,g-1$ and then takes two letters $b$
through $p_r$. A detour through $P$ cannot shorten either route.

For $t\le M-1$, correction suffixes have length at most $g-2$,
so no positive correction reaches color $g-1$, whose trajectory count is at most one.
Suppose a word of length at most $M$ has a first collision.
For a colliding $b$, the distinct source states must lie in its only
non-singleton fiber, all of color $g-1$, contradicting this count bound.
For a colliding $a$, the sources must be $u_h,u_M$, its only non-singleton fiber.
They are not both initially present. A preceding $b$ cannot give $u_M$,
which has no $b$-preimage. A preceding $a$ requires sources
$u_{h-1},u_{M-1}$ one step earlier, both of color $g-1$, again a contradiction;
for $h=1$, $u_h$ has no $a$-preimage at all.
Thus every word of length at most $M$ preserves $s$ states.
For $t\le M$, every correction suffix has length at most $g-1$,
so no negative correction reaches $A$. Its singleton cells retain multiplicity
at least one. Every such word is safe.

The word $a^{M-g}b^{g+1}$ has length $M+1$.
Exactly its first $b$ has a trajectory at $u_M$, with multiplicity one.
Its negative correction is $e_1T_b^g=e_{p_{r-1}}$, whereas
its positive correction is $e_gT_b^g=e_{p_r}$.
It therefore omits $p_{r-1}$ and attains equality, also when $r=1$.

For synchronization put $c=a^{M-h}b^s$.
The $a$-run fixes the cyclic tail $u_{h+1},\ldots,u_M$.
For $j<h$, it sends $u_j$ to $u_{M-h+j}$, which has the same remaining
$b$-distance $h-j$ to $p_r$ as $u_j$.
For $j=h$, it sends $u_h$ to $u_M$, which reaches $p_0$ in $g+1+r=s$ letters $b$.
Consequently $c$ fixes $S\setminus\{u_h\}$ and sends $u_h$ to its $b$-predecessor $p_0$.
After $c$, each $bc$ extends the missing interval by one point in forward
$b$-order. Thus $b^{g+1}c(bc)^{s-2}$ resets $Q$.

\emph{The boundary $M=2g$.}
Let $a$ fix $P$ and cycle $U$.
Define $p_0b=u_g$, $p_ib=p_{i-1}$ and
\[
 u_jb=\begin{cases}
 p_r&j\equiv g-1\pmod g,\\
 u_{j+1}&j<M,\ j\not\equiv g-1\pmod g,\\
 u_M&j=M.
 \end{cases}
\]
Now $b$ cycles $B=P\cup\{u_g,\ldots,u_{2g-1}\}$ and fixes $u_{2g}$;
$Qb^{g+1}=E=B\cup\{u_{2g}\}$.
Put $d=a^{g-1}b^s$. It fixes $P$ and on $E$ has actions
\[
 u_gd=u_{2g-1},\quad u_{g+1}d=u_{2g},\quad
 u_{g+j}d=u_{g+j-1}\ (2\le j\le g-1),\quad u_{2g}d=u_{2g-1}.
\]
Therefore $Eda^g=S$, proving actual reachability. The two cycles prove strong connectivity.
Assigning $u_j$ color $1+((j-1)\bmod g)$ gives exactly
\eqref{eq:error}--\eqref{eq:expand} again.
The same delay and collision arguments apply, with $a$ now a permutation.
The word $a^gb^{g+1}$ attains avoidance: only the first negative correction
has a suffix of $g$ letters $b$ and reaches $p_{r-1}$; no positive correction reaches $A$.
Finally $e=d^g$ fixes $E\setminus\{u_g\}$ and sends $u_g$ to $u_{2g}$.
Each $be$ removes one further point of $B$ while retaining $u_{2g}$.
Hence $b^{g+1}e(be)^{s-1}$ resets $Q$.
\end{proof}
Taking $g=s-r-1$ and $M=n-r-1$ proves part~(3) of Theorem~\ref{thm:sharp-faces}.

\enlargethispage{2\baselineskip}
\section*{Acknowledgments}
We thank Marek Szyku\l a for correspondence on Open Problem~4 and subset compression.

\section*{Data and code availability}
The implementation and reproduction instructions are available at
\href{https://github.com/dongxuelian2/cerny-direct-attack/tree/9da41353adb79b495e13aa4e4300297371452c26}{repository revision \texttt{9da4135}}~\cite{code}.
The accompanying \texttt{docs/RUNNING.md} records input partitions,
commands and completion records, distinguishing full calculations from
bounded validation runs.

\section*{Declaration of generative AI and AI-assisted technologies in the manuscript preparation process}
During the preparation of this work, the authors used OpenAI ChatGPT for
exploratory technical discussions, assistance with code development,
LaTeX preparation, and English-language drafting and editing. The authors
reviewed and edited the resulting material and take full responsibility
for the mathematical arguments, software, computational results, and
content of the manuscript.

\bibliographystyle{plain}
\bibliography{references}
\end{document}